%% file: Simsek_ArXiv.tex
\documentclass[a4paper,11pt,twoside,onecolumn,final]{article}
\usepackage{amsmath,amsfonts,amssymb,amscd,mathrsfs}
\numberwithin{equation}{section}

\usepackage{extarrows}
\usepackage{epsfig}
\usepackage{enumerate}
\usepackage{color}
\usepackage{graphicx}
\graphicspath{{Figures/}}

\usepackage{psfrag}
\usepackage[nodayofweek]{datetime}
\usepackage[paper=a4paper,hscale=0.655,vscale=0.7618,centering,heightrounded]{geometry}
\usepackage[ntheorem]{empheq} 
\empheqset{box=\bigfbox} 
\usepackage[thmmarks,amsmath,hyperref]{ntheorem} 
\theoremstyle{plain} 
\theoremheaderfont{\normalfont\bfseries} 
\theorembodyfont{\slshape}
\theoremsymbol{\ensuremath{_\Box}}
\theoremnumbering{arabic}
\newtheorem{theorem}{Theorem}[section]

\theoremnumbering{arabic}
\newtheorem{proposition}{Proposition}[section]

\theoremstyle{plain}
\theoremheaderfont{\normalfont\bfseries} 
\theorembodyfont{\upshape}
\theoremsymbol{\ensuremath{_\Box}}
\newtheorem{remark}[proposition]{Remark}  

\theoremstyle{nonumberplain}
\theoremheaderfont{\normalfont\bfseries} 
\theorembodyfont{\upshape}
\theoremsymbol{\ensuremath{_\blacksquare}} 
\newtheorem{proof}{Proof}
\theoremsymbol{\ensuremath{_\blacksquare}}

\theoremlisttype{allname}
\allowdisplaybreaks[3] 
\definecolor{darkgray}{gray}{0.4}  
\definecolor{ddarkgray}{gray}{0.2} 
\definecolor{redgray}{rgb}{0.5,0.25,0.25}  
\definecolor{bluegray}{rgb}{0.25,0.25,0.5}  
\usepackage[superscript,biblabel]{cite}
\usepackage[small,bf]{caption} 
\usepackage[notref,notcite,draft]{showkeys}

\definecolor{refkey}{rgb}{0,1,0} 
\definecolor{labelkey}{rgb}{0,1,1}
  \renewenvironment{thebibliography}[1]{%
    \begin{oldthebibliography}{#1}%
       \addtolength{\itemsep}{-0.5ex}%
  }%
  {%
    \end{oldthebibliography}%
  }
\title{\textbf{\Large \parbox{\linewidth}{\centering
Diffuse-Interface Two-Phase Flow Models with\\[-.125\baselineskip] Different Densities: A New Quasi-Incompressible\\[-.125\baselineskip] Form and a Linear Energy-Stable Method
}
}}

\author{%
\parbox{\linewidth}{\centering
{\normalsize M.~Shokrpour Roudbari$^{a,\dag}$, G.~\c{S}im\c{s}ek$^{a,\ddag,*}$, E.H.~van Brummelen$^{a,\S}$, 
\\ and K.G.~van der Zee$^{b}$}
~\\[.5\baselineskip]
\textit{\small $^{a}$Eindhoven University of Technology, Multiscale Engineering Fluid Dynamics,\\
5600 MB, Eindhoven, Netherlands \\
$^{\dag}${}m.shokrpour.roudbari@tue.nl, $^{\ddag}${}grkemsimsek@gmail.com, $^{\S}${}e.h.v.brummelen@tue.nl
~\\[.5\baselineskip]
$^{b}$The University of Nottingham, School of Mathematical Sciences,\\ University Park, Nottingham, NG7 2RD, United Kingdom\\
kg.vanderzee@nottingham.ac.uk
}}}

\date{\small 7 June 2017}

\usepackage{framed}
\usepackage{float,subfig}

\begin{document}
\maketitle

\renewcommand*{\thefootnote}{\fnsymbol{footnote}}
\footnotetext[1]{Corresponding author}
\renewcommand*{\thefootnote}{\alph{footnote}}


\begin{abstract}
\noindent
While various phase-field models have recently appeared for two-phase fluids with different densities, only some are known to be thermodynamically consistent, and practical stable schemes for their numerical simulation are lacking. In this paper, we derive a new form of thermodynamically-consistent quasi-incompressible diffuse-interface Navier--Stokes Cahn--Hilliard model for a two-phase flow of incompressible fluids with different densities. The derivation is based on mixture theory by invoking the second law of thermodynamics and Coleman--Noll procedure. We also demonstrate that our model and some of the existing models are equivalent and we provide a unification between them. In addition, we develop a linear and energy-stable time-integration scheme for the derived model. Such a linearly-implicit scheme is nontrivial, because it has to suitably deal with all nonlinear terms, in particular those involving the density. Our proposed scheme is the first linear method for quasi-incompressible two-phase flows with nonsolenoidal velocity that satisfies discrete energy dissipation independent of the time-step size, provided that the mixture density remains positive. The scheme also preserves mass. Numerical experiments verify the suitability of the scheme for two-phase flow applications with high density ratios using large time steps by considering the coalescence and break-up dynamics of droplets including pinching due to gravity.
%
\\[.75\baselineskip]
\\[.75\baselineskip]

\end{abstract}

\input{StandardSetOfCommands.tex}
\newcommand*\widefbox[1]{\fbox{\hspace{0.05em}#1\hspace{0.05em}}}

\newenvironment{newalgorithm}[1][hbp]
  {\renewcommand{\tablename}{Algorithm} 
   \begin{table}[]%
  }{\end{table}}

\section{Introduction}
\label{introduction}
Diffuse-interface (phase-field) models have emerged as a reliable and versatile alternative to sharp-interface methods for multi-phase flows. The main idea of diffuse-interface models is to replace the sharp interface by a thin, but finite, transition region, where a partial mixing of macroscopically immiscible fluids is allowed, and to define a continuous order parameter (the phase variable) using mass or volume concentration. This leads to the main advantage of diffuse-interface models, which is the natural capability of capturing topological changes, e.g., break-up and coalescence. Another fundamental feature of diffuse-interface models is their rigorous thermodynamic basis. Most of the models satisfy a nonlinear stability relationship such as dissipation of a non-convex free-energy functional, which endows these models with a strong mathematical foundation. Over the years, diffuse interface models have been analyzed theoretically{}\cite{EllPROC1989,EllZheARMA1986,GurPD1996,AbeCMP2009} and used widely in many applications{}\cite{AndMcfWhePD2000,BoyLapTPM2010,KimLowIFB2005,OdeHawPruM3AS2010,Brummelen:2016qa}, while several reviews\cite{AndMcfWheARFM1998,GomZeeBOOK-CH2016,KimCCP2012} have appeared of phase-field models in the context of fluid mechanics. 
\par
In this work we derive a new form of diffuse-interface model, and a corresponding practical time-stepping method, for binary-fluid flows whose components are incompressible with \emph{different} densities. The model is a so-called quasi-incompressible\cite{LowTruPRSLA1998} Navier--Stokes-Cahn--Hilliard (NSCH) model, which is based on the Navier--Stokes equations coupled with the convective Cahn--Hilliard equation. The Cahn--Hilliard equation is a fundamental continuum model for phase separation individually, which was introduced by Cahn~\& Hilliard\cite{CahHilJChP1958} in~1958. Since it is a fourth-order singularly-perturbed nonlinear parabolic PDE, it is challenging to solve it numerically. Coupling it with the Navier--Stokes equations increases the mathematical complexity of the model which makes it difficult to design provably stable numerical schemes. We will introduce a linear and stable time-integration scheme for our derived quasi-incompressible NSCH model. The scheme preserves the structural properties of the model, viz. mass conservation and energy dissipation, at the semi-discrete level independent of the time-step size.
\par
To explain the principal idea behind our model derivation, let us give a short overview of existing NSCH models. The concepts underlying diffuse-interface models for immiscible binary fluids were introduced in the classical works by Van~der~Waals{}\cite{Waa1893} and~Korteweg{}\cite{KortewegANSES1901}. One of the first appearances of the coupled NSCH model, called `Model~H', can be found in the review by Hohenberg \&~Halperin{}\cite{HohHalRMP1977}. This model assumes a solenoidal mixture-velocity field ($\div\mathbf{v} = 0$) and the coupling takes place via a convective term in the Cahn--Hilliard equation and an additional stress tensor term in the Navier--Stokes equations. A derivation of Model~H in the rational continuum mechanics framework was presented by~Gurtin, Polignone \&~Vi\~{n}als\cite{GurPolVinM3AS1996}, who established compatibility with the second law of thermodynamics using the standard Coleman--Noll procedure\cite{ColNolARMA1963,TruBOOK1984,GurFriAnaBOOK2010}. Importantly, Model~H assumes that the individual densities and the mixture density are uniformly constant. However, that assumption makes Model~H \emph{in}applicable for the large variety of problems with \emph{non-matching component densities}. 
\par
Lowengrub \&~Truskinovsky{}\cite{LowTruPRSLA1998} and Abels, Garcke \&~Gr{\"u}n\cite{AbeGarGruM3AS2012} extended Model~H to thermodynamically-consistent models for \emph{non-matching} densities using two different modelling assumptions on the velocity field (mass averaged and volume averaged velocity, respectively) and on the phase-field variable (mass concentration and volume fraction, respectively). Although the models by Lowengrub \&~Truskinovsky and Abels~\emph{et al.} are developed to represent the same type of flow dynamics, the resulting equations have significant differences due to the underlying modeling choices. More precisely, by adopting a \emph{mass}-averaged mixture velocity, the model by Lowengrub \&~Truskinovsky (see also Kim \&~Lowengrub\cite{KimLowIFB2005}) leads to a (generally) \emph{non-solenoidal} velocity field ($\div\mathbf{v} \neq 0$) and additional nonlinear terms compared to Model~H, whereas the \emph{volume}-averaged mixture-velocity model of Abels~\emph{et al.} has a \emph{solenoidal} velocity field and a modified momentum equation (similar to an earlier model obtained by Boyer\cite{BoyCF2002} through asymptotic arguments). 
\par
Noteworthy are also the recent follow-up works by Shen, Yang \&~Wang{}\cite{SheYanWanCCP2013} and Aki, Dreyer, Giesselmann \&~Kraus{}\cite{AkiDreyM3AS2014}, who independently introduced seemingly different quasi-incompressible NSCH models based on a volume-fraction phase variable and non-solenoidal mass-averaged velocity. Shen~\emph{et al.} obtained their model without reference to a Coleman--Noll procedure, but were able to demonstrate global energy dissipation for their model. The work of Aki~\emph{et al.} actually contains the derivation of a more general Navier--Stokes--Korteweg/Cahn--Hilliard/Allen--Cahn model in the non-isothermal and isothermal case (and also includes phase transition). Their derivation leads to a Korteweg stress tensor term as commonly seen in Navier--Stokes--Korteweg models{}\cite{AndMcfWheARFM1998,HeiMalIJES2010,LiuGomEvaHugLanJCP2013,GomZeeBOOK-CH2016,FreKotARMA2017}. Other simpler models have also been introduced,\cite{DingSpeShuJCP2007,SheYanSISC2010,BoyLapTPM2010} which neglect certain terms in the above-mentioned quasi-incompressible models. These models seem to be \emph{in}consistent with mixture theory and the second law of thermodynamics, and have been studied mostly because of their simpler implied numerical treatment, which is closer to that of the variable-density Navier--Stokes equations\cite{GueQuaJCP2000} and Volume-of-fluid (VOF) method\cite{ScaZalARFM1999}.
%
%
%
%
%
\par
The first objective of this paper is to derive from mixture theory the following new form of quasi-incompressible NSCH model:
\hspace{-1.3cm}
\begin{equation}
\label{intro_NSCH} 
\begin{aligned}
& \dot{\phi} + \phi  \nabla \cdot \mathbf{v} = \nabla \cdot \left(m(\phi) \nabla \mu \right)  \\
& \mu = \frac{\sigma}{\epsilon} \frac{df}{d\phi} - \sigma\varepsilon \Delta \phi  - \frac{p}{\hat{\rho}} \frac{d \hat{\rho}}{d \phi}  \\
 &\rho \dot{\mathbf{v}} = - \nabla p - \phi \nabla \mu  + \frac{p}{\hat{\rho}} \frac{d \hat{\rho}}{d \phi}  \nabla \phi  +  \nabla \cdot \big( \eta(\phi) ( 2  \mathbf{D}+ \lambda (\nabla \cdot \mathbf{v})\mathbf{I}) \big)- \hat{\rho} g {\boldsymbol{\jmath}}  \\
 & \nabla \cdot \mathbf{v} = \alpha \nabla \cdot (m(\phi)\nabla \mu)  \\
 & \dot{\rho} + \rho \nabla \cdot  \mathbf{v}= 0 \\
&\hat{\rho}= \rho_1 \frac{1+\phi}{2} + \rho_2 \frac{1-\phi}{2}
\end{aligned} 
\end{equation} 
and to prove its thermodynamic consistency. Here, $\phi$~is the volume fraction as phase variable, $\mathbf{v}$ is the mass-averaged mixture velocity, $\mu$ is the chemical potential, $p$ is the pressure and $\rho = \hat{\rho} \equiv \hat{\rho}(\phi)$ is the mixture density. Moreover, $\rho_1$ and $\rho_2$ are the component densities (assumed constant), $\sigma$~is a constant related to the fluid-fluid surface tension, $\varepsilon$~is the interface thickness parameter, $\eta(\phi)$ is the mixture viscosity, $m(\phi)$ is a mobility function, $\alpha= \frac{\rho_2-\rho_1}{\rho_2+\rho_1}$ is a constant,
and $\lambda\geq{-}2/d$ is a constant ($d$~being the dimension). Furthermore, the term $\hat{\rho} g {\boldsymbol{\jmath}}$ stands for the gravitational force and $\mathbf{D} $ is the symmetric velocity gradient tensor. Note that we presented model~(\ref{intro_NSCH}) with a compatible dual-density form, viz. $\rho$ and $\hat{\rho}$. The appearance of two equivalent but distinct representations of the mixture density specifically serves in the construction of a linear energy-stable time-integration scheme; see below.
\par
The distinguishing feature of model~(\ref{intro_NSCH}) is the term $\frac{p}{\hat{\rho}} \frac{d \hat{\rho}}{d \phi}$ (in~(\ref{intro_NSCH})$_2$ and (\ref{intro_NSCH})$_3$), which has not appeared in previous quasi-incompressible NSCH models. We will show that our model is nevertheless equivalent to the model by Shen, Yang and Wang\cite{SheYanWanCCP2013}, as well as the corresponding model in Aki, Dreyer, Giesselmann and Kraus\cite{AkiDreyM3AS2014}. We thereby provide a unification of the existing models based on a mass-averaged velocity and volume-fraction phase variable. 
It should be noted that in the case of matching densities, all of the aforementioned NSCH models for variable-density two-phase flows, including our model, reduce to Model~H. 
\par
\par
The development of stable and efficient time-integration methods for NSCH systems with non-matched densities is challenging on account of the strong coupling between the equations and the various nonlinear terms, in particular those involving the density. An important notion in the analysis of stability of time-integration schemes is that of {\em energy stability\/}, which implies that the discrete time-integration scheme inherits the fundamental free-energy dissipation property of the underlying PDE system.\cite{EyrPROC1998,GomZeeBOOK-CH2016} Unconditional energy stability is crucial for  robustness of diffuse-interface simulations and for proper resolution of interfaces as it enables, in principle, arbitrary time and space discretizations and hence provides a basis for adaptive refinement. 
\par
In the case of the Cahn--Hilliard equation, various energy-stable semi-discrete (continuous in space, discrete in time) schemes of first and second order have appeared over the past years,\cite{WisWanLowSINUM2009,HuWisWanLowJCP2009,GomHugJCP2011,TieGuiACME2015,WuZwiZee2014,GomZeeBOOK-CH2016} some of which are \emph{linear}, i.e., they require only one solution of a linear-algebraic system per time step.\footnote{Linear energy-stable schemes typically require that the nonconvex free energy has \emph{quadratic} growth (possibly obtained by \emph{truncating} super-quadratic growth).\cite{GomZeeBOOK-CH2016}} These schemes have been extended to NSCH systems with matched densities\cite{KimKanLowJCP2004,HanWangJCP2015} and for quasi-incompressible NSCH systems with a solenoidal mixture-velocity field\cite{DingSpeShuJCP2007,SheYanSISC2010,MinNMPDE2013,SheYanSINUM2015,GarHinKahAPNUM2016}. However, \emph{non}-solenoidal quasi-incompressible NSCH systems have only received scant consideration so far. The reason for this is that existing techniques for solenoidal systems can not be straightforwardly extended to non-solenoidal systems (which apart from being non-solenoidal also have auxiliary pressure terms). One recent scheme is by Guo, Lin \&~Lowengrub\cite{GuoLinJCP2014}, who proposed a complicated nonlinear energy-stable scheme for the model by Lowengrub \&~Truskinovsky\cite{LowTruPRSLA1998}. 
%
\par
\par
The second objective of this work is to introduce a simple {\em linear\/} energy-stable time-integration scheme for model~\eqref{intro_NSCH}. The proposed scheme is energy stable independent of the time-step size, provided that the mixture density remains positive. Moreover, under standard assumptions on the boundary conditions, the scheme is mass conservative. 
To our knowledge, this is the first work with a linear energy-stable time-integration scheme for a quasi-incompressible NSCH system with non-solenoidal velocity field. In contrast to existing time-integration schemes, we make use of a dual-density formulation involving both~$\rho$ and~$\hat{\rho}$ in its discretization (see Section~\ref{Section3_2:en_dis_time}). The dual-density formulation provides the basis for the energy stability of the scheme, especially for numerical computations with high component-density ratios.
\par
The remainder of this paper is structured as follows: In Section~\ref{section2}, we derive the new form of thermodynamically consistent quasi-incompressible model for two-phase flows and discuss its equivalence to other models. Section~\ref{section3} presents a weak form of the system and a time-discrete scheme including proofs of continuous and discrete energy dissipation. In Section~\ref{section4}, we exhibit the properties of our fully-discrete scheme based on numerical computations using standard finite elements as a spatial discretization. Section~\ref{section5} presents concluding remarks. 
%
%
\section{Derivation of the Quasi-Incompressible NSCH Model}
\label{section2}

In this section, we present a new form of quasi-incompressible diffuse--interface NSCH model with gravity, inspired by Aki, Dreyer, Giesselmann and Kraus\cite{AkiDreyM3AS2014}. 
We derive the model for an isothermal mixture with a thin interfacial region between two immiscible and incompressible fluids. The derivation is based on the theory of mixtures\cite{BowBOOK-CH1976,TruBOOK1984}, which assumes that the mass, momentum and energy are conserved at the constituent level as well as within the mixture, see Sections~\ref{cont_theory}--\ref{2ndaxiom}. The model is derived in Sections~\ref{colemanNoll}--\ref{choice} via application of the standard Coleman--Noll procedure\cite{ColNolARMA1963,GurFriAnaBOOK2010} to the energy dissipation inequality, which enforces the second law of thermodynamics and endows the model with thermodynamic consistency. We demonstrate the equivalence of the model to other existing models in Section~\ref{sec:Equiv}, and present a non-dimensionalisation and interface-profile analysis in Sections~\ref{sec:nondim}--\ref{sec:ffinterface}.
%
\subsection{Continuum theory of mixtures} \label{cont_theory}
To provide a setting for our model, we consider an open bounded domain $\Omega \subset \mathbb{R}^d$ ($d = 2,3$) and label the fluids of the binary mixture by $k =1,2$. Let the \textit{volume fraction}{}\footnote{The volume fraction and mass concentration are defined by division of volume and mass of the fluids and the mixture. However, their actual definitions are the limit of the volume and the mass of species per unit volume and mass, i.e. $\phi_{k}  = \lim_{V \to 0 } \frac{V_{k}}{V} $ and $c_{k} = \lim_{V \to 0 } \frac{M_{k}}{M} $, respectively. Here, we drop the limit notation for simplicity.}  be $\phi_k = V_k/V$, where $V_k$ is the volume of the $k$th component and $V$ is the total volume of the mixture. Similarly, the \textit{mass concentration} can be defined as $c_k=M_k/M$ with the mass of the $kth$ component, $M_k$, and the total mass of the mixture, $M$. Particularly, since $M = \sum_{k} M_k$ and $V = \sum_{k} V_k$, we have $\sum_{k} c_k = \sum_{k} \phi_k = 1$.
\par
Each component of the fluid is assumed to be incompressible, so we describe a \textit{constant} specific density $\rho_k = M_{k}/V_{k}$ for $k=1,2$. Similarly, \textit{partial} mass densities can be defined as $\tilde{\rho}_{k} = M_k/V$, which are \textit{non-constant} and their sum is the total mass density of the mixture, $ \rho = \sum_{k} \tilde{\rho}_k$. Note that the three densities are related by 
\begin{equation}
\tilde{\rho}_{k} = \rho_k \phi_k = \rho c_k. \label{fraction_relation}
\end{equation} 
\par
Unlike the individual components, the mixture is not assumed to be incompressible. We assume the mixture density to be a function of an order parameter $\phi$ such that $\rho \equiv \hat{\rho}(\phi)$, where $\phi$ is related to the volume fractions according to $\phi := \phi_1 -\phi_2$. It is to be noted that the total density can generally be chosen as a function of various combinations of volume or mass fractions, some of which are mentioned by Abels, Garcke \&~Gr\"un\cite{AbeGarGruM3AS2012}. We define the order parameter $\phi$ based on the volumes $V_1, V_2$ and $V$ by $V_1/V = (1+\phi)/2$ and $V_2/V = (1-\phi)/2$, which leads to the relations $ \tilde{\rho_1} =  \rho_1(1 + \phi)/2$ and $\tilde{\rho_2} =  \rho_2(1 - \phi)/2$, for $\phi \in [-1, 1]$. This implies the following algebraic equation for the mixture density:
\begin{equation}
\label{rho_algebraic}
\displaystyle
 \tilde{\rho}_1 + \tilde{\rho}_2 = \rho_1 \frac{1+ \phi}{2} + \rho_2 \frac{1 - \phi}{2} =: \hat{\rho}(\phi). 
\end{equation}
Note that $\hat{\rho}(\phi)$ coincides with the mixture density $\rho$, where both represent the non-constant mixture density. Although $\rho$ and $\hat{\rho}:= \hat{\rho}(\phi)$ coincide, they serve separate roles in the formulation and, accordingly, can not be used interchangeably. We introduce the mixture velocity $\mathbf{v}$ as the \textit{mass averaged/barycentric} velocity:
\begin{equation}
 \mathbf{v} = \frac{1}{\rho} \sum\limits_k {\tilde{\rho}}_{k} \mathbf{v_{k}}, \label{mass_avgd}
 \end{equation}
where $\mathbf{v}_{k}$ denotes the velocity of component $k$.  
%
\subsection{Balance equations} \label{balance}
Balance of mass for the individual constituents $k=1, 2$ can be written as
\begin{equation}
\partial_{t} \tilde{\rho}_k + \nabla \cdot (\tilde{\rho}_k \mathbf{v}_{k}) = 0, \label{mass}
\end{equation} 
where we assume no mass production of or conversion between the constituents, which is reflected in the zero right hand side. 
%
\par
We replace $\tilde{\rho}_1 $ and $\tilde{\rho}_2$ by $ \rho_1(1 + \phi)/2$ and $ \rho_2(1 - \phi)/2$ in \eqref{mass}, respectively, use the definition of diffusion velocity according to $\displaystyle {\mathbf{w}}_{k}= {\mathbf{v}}_{k} - \mathbf{v}$ and relation \eqref{fraction_relation} to obtain
\begin{align}
&\partial_{t} \left(\frac{\rho_1}{2} ( 1+\phi) \right) + \nabla \cdot \left(\frac{\rho_1}{2} ( 1+\phi) \mathbf{v}\right) + \nabla \cdot \left( \tilde{\rho}_{1} {\mathbf{w}}_{1}\right) = 0 \label{mass_1} \\
&\partial_{t} \left(\frac{\rho_2}{2} ( 1-\phi) \right) + \nabla \cdot \left(\frac{\rho_2}{2} ( 1-\phi) \mathbf{v}\right) + \nabla \cdot \left( \tilde{\rho}_{2} {\mathbf{w}}_{2}\right) = 0  \label{mass_2}
\end{align}
Multiplying \eqref{mass_1} and \eqref{mass_2} by $1/\rho_1$ and $1/\rho_2$, respectively, and subtracting the resulting equations gives the phase equation:
\begin{equation}
\dot{\phi} + \phi \nabla \cdot \mathbf{v} + \nabla \cdot \mathbf{h}= 0 \label{phase_eq}
\end{equation}
where $ \displaystyle \mathbf{h} = \phi_1 {\mathbf{w}}_{1} - \phi_2 {\mathbf{w}}_{2}$ is the mass flux due to diffusive changes in the phase variable $\phi$ and $\displaystyle  \dot{x}$ is the material derivative $\displaystyle  \dot{x} = \partial_{t} x + \mathbf{v} \cdot \nabla x$ for any field variable $x$. 
\par
Similarly, summing the equations in \eqref{mass_1}--\eqref{mass_2} and using the identity $\rho_1 \phi_1 {\mathbf{w}}_{1} +\rho_2 \phi_2 {\mathbf{w}}_{2}=0$, which can be inferred from \eqref{mass_avgd}, we obtain the following quasi-incompressibility relation for the mixture velocity $\mathbf{v}$:
\begin{equation}
\nabla \cdot \mathbf{v} + \alpha\,  \nabla \cdot \mathbf{h} = 0, \label{quasi_incomp}
 \end{equation}
where 
\begin{equation}
\alpha := \frac{\rho_2 - \rho_1}{\rho_2+\rho_1} \label{alpha}
\end{equation}
is a constant. By combining \eqref{phase_eq} and \eqref{quasi_incomp} we obtain the following relation between the phase variable and the mixture velocity:
\begin{equation}
\nabla \cdot \mathbf{v} = \frac{\alpha}{1-\alpha \phi} \dot{\phi}. \label{vel_phase_rel}
\end{equation}
\begin{remark}
If the specific densities are equal, that is,  ${\rho}_1 ={\rho}_2$, then \eqref{quasi_incomp} 
reduces to $\nabla \cdot \mathbf{v} = 0$. The barycentric velocity is therefore solenoidal if the specific densities of the two components coincide, but not generally otherwise.
\end{remark}
Balance of mass is not only satisfied for the components, but also for the mixture. Indeed, summing \eqref{mass} in $k$, we obtain
\begin{equation}
\partial_{t} \rho + \nabla \cdot (\rho \mathbf{v}) = 0 \quad \text{or} \quad \dot{\rho} + \rho \nabla\cdot\mathbf{v} = 0 \label{mass_global}
\end{equation} 
and similarly
\begin{equation}
\partial_{t} \hat{\rho} + \nabla \cdot ( \hat{\rho}\mathbf{v}) = 0, \label{mass_global_dens_alg}
\end{equation} 
where $\hat{\rho}$ is the algebraic definition of the mixture density as in \eqref{rho_algebraic}. Then, from \eqref{mass_global_dens_alg} it follows that
\begin{equation}
\nabla \cdot \mathbf{v} = - \frac{1}{\hat{\rho}} \dot{\hat{\rho}} = - \frac{1}{\hat{\rho}} \frac{d \hat{\rho}}{d \phi} \dot{\phi}, \label{velocity_div}
\end{equation} 
which is another identity for the divergence of the mixture velocity in \eqref{quasi_incomp}. Additionally, from \eqref{vel_phase_rel} and \eqref{velocity_div}, we have
\begin{equation}
 - \frac{1}{\hat{\rho}} \frac{d \hat{\rho}}{d \phi} =  \frac{\alpha}{1-\alpha \phi}, \label{relation2}
 \end{equation}
 where $ d \hat{\rho} /d \phi$ is constant.
 \par
For our model, we are interested in the velocity of the mixture, $\mathbf{v}$, more than the specific velocities, ${\mathbf{v}}_k$. Therefore, instead of introducing the momentum equation for each constituent, we will consider the mixture momentum balance
\begin{equation}
\partial_t (\rho \mathbf{v}) + \nabla \cdot \left( \rho \mathbf{v} \otimes \mathbf{v}\right) = \nabla \cdot \mathbf{T} +\mathbf{b}, \nonumber
\end{equation}
where $\mathbf{T}$ is the stress tensor of the mixture and $\mathbf{b}$ is the body force. By \eqref{mass_global}, this can be simplified to
\begin{equation}
\rho \dot{\mathbf{v}}= \nabla \cdot \mathbf{T} + \mathbf{b}. \label{momentum}
\end{equation}
We restrict our considerations here to body-force terms corresponding to gravitational effects. Accordingly, 
$\mathbf{b}$ in~\eqref{momentum} is replaced by $-\hat{\rho} g {\boldsymbol{\jmath}} $ with $g$ as 
gravitational acceleration and ${\boldsymbol{\jmath}} $ as the vertical unit vector.

\begin{remark}
An alternative mixture velocity definition to~\eqref{mass_avgd} is the \textit{volume-averaged} velocity $\displaystyle \mathbf{v} = \phi_{1} {\mathbf{v}}_{1} + \phi_{2} {\mathbf{v}}_{2}$, which is employed in other works\cite{AbeGarGruM3AS2012,BoyCF2002,GarHinKahAPNUM2016}. This choice would reduce~\eqref{quasi_incomp} to ${\nabla \cdot \mathbf{v}} = 0$ even for non-matching densities, i.e. for ${\rho}_1 \neq {\rho}_2$. However, this simplification requires a change in the conservation equations for mass and momentum to obtain a thermodynamically consistent model. More explicitly, volume averaged velocity leads to an additional term in the momentum conservation equation related to diffusion of components in order to describe the density flux.\cite{AbeGarGruM3AS2012}
\end{remark}
%
\subsection{The second axiom of thermodynamics} \label{2ndaxiom}
Similar to the momentum equation \eqref{momentum}, we assume the balance of internal energy not for the individual constituents, but for the mixture. However, instead of introducing the balance of energy equation here, we directly consider the second axiom of thermodynamics in the form of an energy dissipation inequality.\cite{GurPolVinM3AS1996} The connection between the internal-energy density and the dissipation inequality is made via the Helmholtz free-energy density, $\rho \psi$. 
\par
Let us consider the following energy dissipation inequality in terms of the free energy $\rho \psi$, the kinetic  and gravitational potential energy $K({\mathcal{P}})$ and $G({\mathcal{P}})$, the total work done by macro- and micro-stresses $ W({\mathcal{P}}) $ and the energy transported by diffusion $ M({\mathcal{P}}) $:
\begin{equation}
\frac{d}{d t} \int_{{\mathcal{P}}(t)} \Big(\rho \psi  + K({\mathcal{P}}) + G({\mathcal{P}})\Big) \, dv \leq W({\mathcal{P}}) - M({\mathcal{P}}), \label{dissipation}
\end{equation}
where $\mathcal{P}$ denotes any time-dependent subset of $\Omega$, which moves with the mixture velocity $\mathbf{v}$.
\par
Equation \eqref{dissipation} implies that the work on ${\mathcal{P}}$ plus the rate at which free energy is transported to ${\mathcal{P}}$ by diffusion always exceed the sum of the increase in free, kinetic and the gravitational potential energy. Specifically, following Gurtin\cite{GurPD1996} we consider
\begin{align}
& \hspace{3cm} K({\mathcal{P}}) = \frac{1}{2} \rho |\mathbf{v}|^2,\quad G({\mathcal{P}}) = \hat{\rho}g y, \label{eqn_def_1} \\
& W({\mathcal{P}}) = \int_{\partial \mathcal{P}} \mathbf{T} \mathbf{n}\cdot \mathbf{v} \, da + \int_{\partial  {\mathcal{P}}} \dot{\phi}\boldsymbol\xi\cdot\mathbf{n}\, da, \quad 
 M({\mathcal{P}}) = \int_{\partial {\mathcal{P}}} \mu \mathbf{h}\cdot \mathbf{n}\, da, \label{eqn_def_2}
\end{align}
 where $y$ is the vertical coordinate, $\mathbf{n}$ is the exterior unit normal vector to the boundary of $\mathcal{P}$, $\partial {\mathcal{P}}$, $\mu$ denotes a chemical potential and $\boldsymbol\xi$ is a vectorial surface force as a component of micro-stress.
\par
Upon substituting \eqref{eqn_def_1} and \eqref{eqn_def_2} into \eqref{dissipation} and applying the Reynolds transport theorem, we obtain
\begin{equation}
\label{local_dissipation-1}
\begin{aligned}
&\int_{{\mathcal{P}}} \frac{\partial}{ \partial t} (\rho \psi) \, dv  +   \int_{\partial \mathcal{P}} \rho \psi    \mathbf{v} \cdot\mathbf{n}\, da + \int_{{\mathcal{P}}}  \frac{1}{2} \frac{\partial}{\partial t} (\rho |\mathbf{v}|^2 )  \, dv  \\
& +\int_{\partial \mathcal{P}}  \frac{1}{2}   \rho |\mathbf{v}|^2 \mathbf{v} \cdot\mathbf{n}\, da + \int_{{\mathcal{P}}}   \partial_{t} \hat{\rho} g y  \, dv + \int_{\partial \mathcal{P}}  \hat{\rho} g y   \mathbf{v} \cdot\mathbf{n}\, da  \\
& \leq    \int_{\partial \mathcal{P}} \mathbf{T} \mathbf{n}\cdot \mathbf{v} \, da + \int_{\partial  {\mathcal{P}}} \dot{\phi}\boldsymbol\xi\cdot\mathbf{n}\, da - \int_{\partial {\mathcal{P}}} \mu \mathbf{h}\cdot \mathbf{n}\, da. 
\end{aligned}
\end{equation}
Then using the conservation laws \eqref{mass_global}, \eqref{mass_global_dens_alg} and \eqref{momentum} on the left hand side and the divergence theorem on the right hand side of the inequality,  the following local form of \eqref{local_dissipation-1} is obtained:
\begin{equation}
\rho \dot{\psi}- \mathbf{T}:\nabla \mathbf{v} - \nabla \cdot ( \dot{\phi} \boldsymbol\xi ) + \nabla \cdot (\mu \mathbf{h}) \leq 0. \label{local_dissipation}
\end{equation}
By applying the product rule to the last two terms in \eqref{local_dissipation} and introducing
\begin{equation}
\dot{\overline{\rho \psi}} = \rho \dot{\psi}  + \psi \dot{\rho}  \quad \text{and} \quad \mathbf{T} = \mathbf{T}_{0} - p\mathbf{I}, \nonumber
\end{equation}
we obtain 
\begin{equation}
\dot{\overline{\rho \psi}} - \psi \dot{\rho} - \mathbf{T}_{0}:\nabla \mathbf{v}  +p \nabla \cdot \mathbf{v} - \dot{\phi} \nabla \cdot \boldsymbol\xi
-\boldsymbol\xi   \nabla (\dot{\phi} )+ \mu \nabla \cdot \mathbf{h} + \mathbf{h}\cdot \nabla \mu \leq 0. \label{local_dissip_2}
\end{equation}
The partition $\mathbf{T} = \mathbf{T}_{0} - p\mathbf{I}$ is such that $- p\mathbf{I}$ corresponds to the hydro-static part of the stress tensor $\mathbf{T}$. The reduced dissipation inequality \eqref{local_dissip_2} forms the basis for determining our class of admissible constitutive relations and the independent variables. 
\begin{remark}
We define $\psi$ as the density of the Helmholtz free energy with respect to the mass measure and, accordingly, $\rho \psi$ as the volumetric free-energy density. In many other formulations\cite{GurPolVinM3AS1996,AbeGarGruM3AS2012}, $\psi$ is defined directly as a volumetric free energy. Our choice for~$\psi$ is crucial to simplify the final derived model such that the momentum equation has a density-free capillary force term, which enables the construction of a linear time-integration scheme.
\end{remark}
%
\subsection{Constitutive equations and Coleman--Noll procedure}
\label{colemanNoll}
We consider the independent variables: $\phi, \mathbf{v}, \mu, p$ and the dependent functions: $\rho \psi, \hat{\rho}, \mathbf{T}, \mathbf{h}, \mathbf{\boldsymbol\xi}$. The explicit form of the functions in terms of independent variables is determined via the Coleman--Noll procedure, except $\hat{\rho}$ which has been defined previously in \eqref{rho_algebraic}. We start with the constitutive assumption for the energy density  $\rho \psi$ according to ${\rho \psi} = \widehat{\rho \psi} (\phi, \grad \phi)$, which gives
\begin{equation}
\dot{\widehat{\rho\psi} } =  \frac{\partial (\widehat{\rho\psi})}{\partial \phi}  \dot{\phi}+ \frac{\partial (\widehat{\rho \psi})}{\partial \nabla \phi}  \dot{\overline{\nabla \phi}}. \label{rhopsi}
\end{equation}
Furthermore, on account of $\dot{\overline{\nabla \phi}} = \nabla (\dot{\phi} ) - \nabla \mathbf{v}\cdot \nabla\phi$, it follows from \eqref{rhopsi} that
 \begin{equation}
\dot{\widehat{\rho\psi} } =  \frac{\partial (\widehat{\rho\psi})}{\partial \phi}  \dot{\phi}+ \frac{\partial (\widehat{\rho \psi})}{\partial \nabla \phi}  \nabla (\dot{\phi} )  - \left( \frac{\partial (\widehat{\rho \psi})}{\partial \nabla \phi} \otimes \nabla\phi \right) :\nabla \mathbf{v} . \label{rhopsi_new}
\end{equation}
Using the identities \eqref{quasi_incomp}, \eqref{mass_global} and \eqref{rhopsi_new}, the inequality \eqref{local_dissip_2} can be recast into
\begin{equation}
\begin{split}
\frac{\partial (\widehat{\rho\psi})}{\partial \phi} \, \dot{\phi}+ \frac{\partial (\widehat{\rho \psi})}{\partial \nabla \phi}  \,\nabla (\dot{\phi} ) &-\left( \frac{\partial (\widehat{\rho \psi})}{\partial \nabla \phi} \otimes \nabla\phi \right) :\nabla \mathbf{v} + (\widehat{\rho\psi})\mathbf{I} :\nabla \mathbf{v} - \mathbf{T}_0 :\nabla \mathbf{v} \\
& \hspace{-1.5cm}+ \boxed{p \nabla \cdot \mathbf{v} }
- \dot{\phi} \nabla \cdot \boldsymbol\xi -\boldsymbol\xi \nabla (\dot{\phi} )  + \boxed{\mu \, \nabla \cdot \mathbf{h} }+ \mathbf{h}\cdot \nabla \mu \leq 0,
\end{split} \label{local_dissip_3}
\end{equation}
where $p\mathbf{I}: \nabla\mathbf{v} =  p\nabla\cdot \mathbf{v} $.
In \eqref{local_dissip_3}, we have written each term as a contraction of a dependent and an independent variable except for the boxed terms. The boxed terms can be recast into the same form by means of the relations 
\begin{equation}
 \nabla\cdot \mathbf{v} = - \frac{1}{\hat{\rho}} \frac{d \hat{\rho}}{d \phi} \dot{\phi} \quad \text{and}  \quad  \nabla \cdot \mathbf{h} =  - \dot{\phi} - \phi \nabla \cdot \mathbf{v}, \label{relation1}
 \end{equation}
which results in the following local dissipation inequality 
\begin{equation}\label{local_dissip_4}
\begin{split}
\frac{\partial (\widehat{\rho\psi})}{\partial \phi} \, \dot{\phi} + \frac{\partial (\widehat{\rho \psi})}{\partial \nabla \phi}  \,\nabla (\dot{\phi} ) &-\left( \frac{\partial (\widehat{\rho \psi})}{\partial \nabla \phi} \otimes \nabla\phi \right) :\nabla \mathbf{v} + (\widehat{\rho\psi})\mathbf{I} :\nabla \mathbf{v} - \mathbf{T}_0 :\nabla \mathbf{v} \\
& \hspace{-2.0cm} - \frac{p}{\hat{\rho}} \frac{d \hat{\rho}}{d \phi}   \dot{\phi} - \dot{\phi} \nabla \cdot \boldsymbol\xi
-\boldsymbol\xi  \nabla (\dot{\phi}) + \mu (- \dot{\phi} - \phi \nabla \cdot \mathbf{v})+ \mathbf{h}\cdot \nabla \mu \leq 0.
\end{split}
\end{equation}
Replacing $\mathbf{T}_{0}$ by $\mathbf{T}_{0} = \mathbf{T} + p\mathbf{I}$ and rearranging terms, we obtain
\begin{equation}
\begin{split}
& -\nabla \mathbf{v} : \left( \mathbf{T} + p \mathbf{I}+ \frac{\partial (\widehat{\rho \psi})}{\partial \nabla \phi} \otimes \nabla\phi + (\mu \phi - \widehat{\rho\psi} )\mathbf{I}\right) \\
&+\dot{\phi} \left(\frac{\partial (\widehat{\rho\psi})}{\partial \phi}  - \nabla \cdot \boldsymbol\xi - \mu - \frac{p}{\hat{\rho}} \frac{d \hat{\rho}}{d \phi}    \right) \\
& +\nabla (\dot{\phi} ) \left( \frac{\partial (\widehat{\rho \psi})}{\partial \nabla \phi} - \boldsymbol\xi  \right)+  \mathbf{h}\cdot \nabla \mu \leq 0.
\end{split}\label{local_dissip_5}
\end{equation}
Based on the inequality \eqref{local_dissip_5}, we choose the constitutive relations as:
\begin{equation}
\begin{split}
\mathbf{{T}} &= \mathbf{\hat{T}} (\phi, \grad \phi, \grad \mathbf{v}) \\
{\boldsymbol\xi} &= \hat{\boldsymbol\xi}(\phi, \grad \phi) \\
\mathbf{h} &= \mathbf{\hat{h}}(\phi, \grad \phi, \mu, \grad \mu) ,
\end{split} \nonumber
\end{equation}
where according to the standard Coleman--Noll argument, the form of the functions can be selected arbitrarily. To avoid that a variation of $\dot{\phi }$ and $\nabla (\dot{\phi} )$ leads to a violation of the inequality \eqref{local_dissip_5}, we insist that:
\begin{align}
& \frac{\partial (\widehat{\rho\psi})}{\partial \phi}  - \nabla \cdot \boldsymbol\xi - \mu - \frac{p}{\hat{\rho}} \frac{d \hat{\rho}}{d \phi}   = 0 \label{equation1_a} \\
& \frac{\partial (\widehat{\rho \psi})}{\partial \nabla \phi} - \boldsymbol\xi  = 0, \label{equation1_b}
\end{align}
which provide equations for $\mu$ and $ \boldsymbol\xi$, respectively. The following constitutive relations then ensure compliance with \eqref{local_dissip_5}:
\begin{align}
 &\mathbf{h} = -m(\phi) \nabla \mu, \label{equation2_b} \\
&\mathbf{T} +p\mathbf{I}+ \left(\frac{\partial (\widehat{\rho \psi})}{\partial \nabla \phi} \otimes \nabla\phi\right)+ (\mu \phi - \widehat{\rho\psi})\mathbf{I} = \eta(\phi) \left(  2 \mathbf{D}+ \lambda (\nabla \cdot \mathbf{v}) \mathbf{I} \right) \label{equation2_a}
\end{align}
where $\mathbf{D} = \tfrac{1}{2}\left(\nabla \mathbf{v} + \nabla \mathbf{v}^T \right)$ is the symmetric velocity gradient tensor and $\nabla \mathbf{v}^T$ denotes transpose of $\nabla \mathbf{v}$.
Then \eqref{local_dissip_5} is satisfied as, in particular,
\begin{equation}
- \eta(\phi) \left(  2 \mathbf{D}+ \lambda  (\nabla \cdot \mathbf{v})\mathbf{I}  \right) : \nabla \mathbf{v}  -m(\phi) |\nabla \mu|^2 \leq 0, \nonumber
\end{equation}
where $\eta(\phi) \geq 0$ is the viscosity, $m(\phi)\geq$ is the mobility. The viscous contribution 
$- \eta(\phi) \left(  2 \mathbf{D}+ \lambda  (\nabla \cdot \mathbf{v})\mathbf{I}  \right)$ is non-positive under the standing assumption $\lambda\ge{}-2/d\,$.\cite{Temam:2001cz} 
%
\begin{remark}
We introduce $ \eta(\phi) \left( 2  \mathbf{D} + \lambda (\nabla \cdot \mathbf{v}) \mathbf{I} \right) $ as the viscous stress tensor where the fluid is considered to be an isotropic Newtonian mixture with volume-fraction-dependent viscosity.\cite{OdeHawPruM3AS2010}
\end{remark}
%
\subsection{Special choice of free energy and the Navier--Stokes Cahn--Hilliard Equation} 
\label{choice}
To obtain a specific quasi-incompressible model, we choose the Helmholtz free energy in the Ginzburg-Landau form\cite{CahHilJChP1958}:
$$\displaystyle \widehat{\rho\psi}(\phi, \nabla \phi ) = \frac{\sigma}{\varepsilon} f(\phi) + \frac{\sigma \varepsilon}{2} |\nabla \phi|^{2},$$
where $\varepsilon >0$ represents the thickness of the interface of the mixture, $\sigma$ is related to the surface energy density and $\displaystyle f(\phi)$ represents a double-well potential, for example the globally $C^{2,1}$-continuous standard (truncated) quartic polynomial according to
\begin{equation}\label{f_phi}
f(\phi) := \left\{ \begin{array}{ll}
 (\phi+1)^2, & \phi <-1\\
 \frac{1}{4} (\phi^2 - 1)^2, & \phi \in [-1, 1]\\
 (\phi-1)^2, & \phi > 1
  \end{array} \right.
\end{equation} 
Then, by the definition of $\widehat{\rho\psi}$ it holds that
\begin{equation}
 \frac{\partial (\widehat{\rho\psi})}{\partial \phi} =  \frac{\sigma}{\varepsilon} \frac{df}{d\phi} ,\qquad \frac{\partial (\widehat{\rho \psi})}{\partial \nabla \phi} = \sigma \varepsilon \nabla \phi. \label{helmoltz_div}
\end{equation}
If we substitute \eqref{helmoltz_div} into \eqref{equation1_a}--\eqref{equation2_a} and by \eqref{phase_eq} and \eqref{momentum}, we obtain the following quasi-incompressible Navier--Stokes-Cahn--Hilliard system:
\begin{subequations}
\label{NSCH_loc}
\begin{align}
&\dot{\phi} + \phi  \nabla \cdot \mathbf{v} = \nabla \cdot \left(m(\phi) \nabla \mu \right) \label{phase_modelII} \\
& \mu = \frac{\sigma}{\epsilon} \frac{df}{d\phi} - \sigma\varepsilon \Delta \phi  - \frac{p}{\hat{\rho}} \frac{d \hat{\rho}}{d \phi} \label{chem_pot_modelII} \\
 &\rho \dot{\mathbf{v}}= - \nabla p - \sigma \varepsilon \nabla \cdot ( \nabla \phi \otimes \nabla \phi) +\nabla \cdot \left( \frac{\sigma}{\varepsilon} f(\phi) + \frac{\sigma \varepsilon}{2} |\nabla \phi|^2 - \mu \phi \right) \mathbf{I} \nonumber \\
 & \hspace{3cm}+ \nabla \cdot \big( \eta(\phi) \left( 2  \mathbf{D}+ \lambda  (\nabla \cdot \mathbf{v})\mathbf{I}  \right) \big) - \hat{\rho} g {\boldsymbol{\jmath}} \label{momentum_modelII} \\
 & \nabla \cdot \mathbf{v} = \alpha \nabla \cdot  (m(\phi)\nabla \mu) \label{div_v_modelII}\\
 & \dot{\rho} + \rho \nabla \cdot  \mathbf{v} = 0 \label{density_alg_modelII_a} \\
 &\hat{\rho}= \rho_1 \frac{1+\phi}{2} + \rho_2 \frac{1-\phi}{2} \label{density_alg_modelII_b}
\end{align}
\end{subequations}
The NSCH system~\eqref{NSCH_loc} is thermodynamically consistent by construction, i.e. it complies with the second law
of thermodynamics.

\begin{remark}
One may note that the two variants \eqref{density_alg_modelII_a} and \eqref{density_alg_modelII_b} of the mixture density both appear in the equation of motion \eqref{momentum_modelII}. More explicitly, the density in \eqref{density_alg_modelII_a} appears in the $\rho \dot{\mathbf{v}}$ term, while the density in \eqref{density_alg_modelII_b} appears in the gravity term $\hat{\rho} g {\boldsymbol{\jmath}}$. \eqref{density_alg_modelII_b} also appears in the equation of chemical potential \eqref{chem_pot_modelII}. Defining the mixture densities $\rho$ and $\hat{\rho}$ via two separate equations in \eqref{density_alg_modelII_a} and \eqref{density_alg_modelII_b} is crucial to obtain an energy dissipative time-discrete scheme.  These two definitions of the mixture density are consistent by construction; see Section~\ref{cont_theory}.
\end{remark}

\begin{remark}
The components of the stress in \eqref{momentum_modelII} can be endowed with specific physical interpretations. The tensor $\eta(\phi) ( 2  \mathbf{D}+ \lambda (\nabla \cdot \mathbf{v}) \mathbf{I} )$ represents the standard viscous stress tensor. The tensors $\sigma \varepsilon  ( \nabla \phi \otimes \nabla \phi) $ and $ \left( \frac{\sigma}{\varepsilon} f(\phi) + \frac{\sigma \varepsilon}{2} |\nabla \phi|^2 - \mu \phi \right) \mathbf{I} $ are associated with capillary forces due to the surface tension and an additional contribution due to quasi-incompressibility, respectively.
\end{remark}
Equation  \eqref{momentum_modelII} can be reformulated such that the additional complicated stress terms assume a simpler form. To this end, we note that for $\displaystyle \sigma \varepsilon \nabla \cdot ( \nabla \phi \otimes \nabla \phi) $ we have the following sequence of identities: 
\begin{align}
\sigma \varepsilon \nabla \cdot ( \nabla \phi \otimes \nabla \phi) &= \sigma \varepsilon (\Delta \phi \nabla \phi + \frac{1}{2} \nabla |\nabla \phi|^2) \nonumber \\
&= \nabla \phi \left (\frac{\sigma}{\varepsilon} \frac{df}{d\phi} - \mu -\frac{p}{\hat{\rho}} \frac{d \hat{\rho}}{d \phi}  \right) + \frac{\sigma \varepsilon}{2} \nabla |\nabla \phi|^2 \nonumber \\
&= \nabla \cdot  \left( \frac{\sigma}{\varepsilon} f(\phi) + \frac{\sigma \varepsilon}{2} |\nabla \phi|^2 - \mu \phi \right)\mathbf{I} + \phi \nabla \mu  -\frac{p}{\hat{\rho}} \frac{d \hat{\rho}}{d \phi} \nabla \phi, \label{simplify}
\end{align}
where the second identity in \eqref{simplify} follows from \eqref{chem_pot_modelII}. Inserting \eqref{simplify} into \eqref{momentum_modelII} leads to the final form of the quasi-incompressible Navier--Stokes Cahn--Hilliard model 
\begin{subequations} \label{final_NSCH}
\begin{empheq}[box=\widefbox]{align}
& \dot{\phi} + \phi  \nabla \cdot \mathbf{v} = \nabla \cdot \left(m(\phi) \nabla \mu \right)  \label{phase_modelII_simp} \\
& \mu = \frac{\sigma}{\epsilon} \frac{df}{d\phi} - \sigma\varepsilon \Delta \phi  - \frac{p}{\hat{\rho}} \frac{d \hat{\rho}}{d \phi} \label{chem_pot_modelII_simp}  \\
 &\rho \dot{\mathbf{v}} = - \nabla p - \phi \nabla \mu  + \frac{p}{\hat{\rho}} \frac{d \hat{\rho}}{d \phi} \nabla \phi  +  \nabla \cdot \left( \eta(\phi) \left( 2  \mathbf{D}+ \lambda (\nabla \cdot \mathbf{v}) \mathbf{I} \right) \right)- \hat{\rho} g {\boldsymbol{\jmath}} \label{momentum_modelII_simp} \\
 & \nabla \cdot \mathbf{v} = \alpha \nabla \cdot (m(\phi)\nabla \mu) \label{div_v_modelII_simp}\\
 & \dot{\rho} + \rho \nabla \cdot  \mathbf{v}= 0 \label{density_alg_modelII_simp_a} \\
& \hat{\rho}= \rho_1 \frac{1+\phi}{2} + \rho_2 \frac{1-\phi}{2}  \label{density_alg_modelII_simp_b}
\end{empheq}
\end{subequations}
\begin{remark}
For matching densities, i.e. $\rho_1 = \rho_2$, equations \eqref{density_alg_modelII_simp_a} and \eqref{density_alg_modelII_simp_b} are trivially satisfied and equations \eqref{phase_modelII_simp}--\eqref{div_v_modelII_simp} reduce to the standard incompressible NSCH system\cite{GurPolVinM3AS1996}. Additionally, the $- \phi \nabla \mu$ term can be rewritten as $\nabla \cdot ( \nabla \phi \otimes \nabla \phi) $ upon redefining the pressure by $\tilde{p} = p - \frac{\sigma}{\varepsilon} f(\phi) - \frac{\sigma \varepsilon}{2} |\nabla \phi|^2 + \mu \phi $.
\end{remark}
\par
In the sequel, we will generally equip \eqref{final_NSCH} with the following natural boundary conditions:
\begin{equation}
\partial_{n} \phi =\partial_{n} \mu  = 0, \quad \mathbf{v} = 0 \quad \text{on}\, \partial \Omega, \label{boun_con}
\end{equation}
where $\partial_{n}(\cdot) $ represents the normal-derivative in the trace sense. Also, for the sake of simplicity, we restrict our considerations to constant viscosity and mobility, i.e. $\eta(\phi):= \eta$ and $m(\phi):= m$.
\begin{remark} Other boundary conditions, e.g. of non-homogeneous form can be considered as well. However, one should take into consideration that there is a compatibility between the velocity $\mathbf{v}$ and the chemical potential $\mu$ due to the quasi-incompressibility equation \eqref{div_v_modelII_simp}, that is,
\begin{equation}
\int_{\partial \Omega} \mathbf{v} \cdot \mathbf{n} \, \, dS = \alpha \int_{\partial \Omega} m(\phi) \, \partial_{n} \mu  \,\, dS \nonumber 
\end{equation}
\end{remark}
\par
The total energy functional associated with \eqref{final_NSCH} is compatible with \eqref{dissipation} and comprises the Helmholtz free energy, the kinetic energy and the gravitational energy:
\begin{equation}
E\left(\phi, \rho, \hat{\rho}, \mathbf{v} \right) := \int_{\Omega} \left( \frac{1}{2} \rho |\mathbf{v}|^2 + \frac{\sigma}{\varepsilon} f(\phi) + \frac{\sigma \varepsilon}{2} |\nabla \phi|^2 + \hat{\rho} g y\right) \, d\Omega, \label{energy}
\end{equation}
Here, $\frac{1}{2} \rho |\mathbf{v}|^2$ is the kinetic energy, $ \frac{\sigma}{\varepsilon} f(\phi) + \frac{\sigma \varepsilon}{2} |\nabla \phi|^2$ is the Cahn--Hilliard free energy and $\displaystyle \hat{\rho} g y$ is the gravitational potential energy.

%
\subsection{Equivalent form and relation to existing models}
\label{sec:Equiv}
Based on the same modeling assumptions, one can derive a seemingly different quasi-incompressible model to~\eqref{final_NSCH}. Let us follow the same steps in the derivation until \eqref{local_dissip_3} and rewrite the term $p \nabla \cdot \mathbf{v}$ using
\begin{equation}
 \nabla\cdot \mathbf{v} = \alpha \dot{\phi} + \alpha \phi \nabla \cdot \mathbf{v} 
 \end{equation}
instead of $\nabla \cdot  \mathbf{v} = -\dfrac{1}{\hat{\rho}}\dfrac{d \hat{\rho}}{d \phi} \dot{\phi}$ in \eqref{relation1}. Then the local dissipation inequality becomes
\begin{equation}\label{local_dissip_alternative}
\begin{split}
&\frac{\partial (\widehat{\rho\psi})}{\partial \phi} \, \dot{\phi} + \frac{\partial (\widehat{\rho \psi})}{\partial \nabla \phi}  \,\nabla (\dot{\phi} ) -\left( \frac{\partial (\widehat{\rho \psi})}{\partial \nabla \phi} \otimes \nabla\phi \right) :\nabla \mathbf{v} + (\widehat{\rho\psi})\mathbf{I} :\nabla \mathbf{v} - \mathbf{T}_0 :\nabla \mathbf{v} \\
&+\alpha  p  \dot{\phi}  +\left(\alpha p \phi\right) \mathbf{I} :\nabla \mathbf{v}  - \dot{\phi} \nabla \cdot \boldsymbol\xi -\boldsymbol\xi \nabla (\dot{\phi} ) + \mu (- \dot{\phi} - \phi \nabla \cdot \mathbf{v})+ \mathbf{h}\cdot \nabla \mu \leq 0.
\end{split}
\end{equation} 
Applying the Coleman--Noll procedure to \eqref{local_dissip_alternative}, we obtain the same equations as \eqref{equation1_b} and \eqref{equation2_b} and the following two equations for chemical potential and the stress tensor:
\begin{equation}
\begin{split}
& \frac{\partial (\widehat{\rho\psi})}{\partial \phi}  - \nabla \cdot \boldsymbol\xi - \mu +\alpha  p = 0\\ 
&\mathbf{T} + p \mathbf{I} + (\boldsymbol\xi \otimes \nabla\phi) + (\mu \phi - \widehat{\rho\psi}+\alpha p \phi)\mathbf{I} =   \eta(\phi) \left( 2  \mathbf{D}+ \lambda (\nabla \cdot \mathbf{v}) \mathbf{I} \right).
\end{split} \label{equation1_alternative}
\end{equation}
Hence, repeating the steps in \eqref{simplify}, we obtain the following alternative quasi-incompressible NSCH model: 
\begin{subequations}\label{modelI_simp}
\begin{empheq}[box=\widefbox]{align}
& \dot{\phi} + \phi  \nabla \cdot \mathbf{v}  = \nabla \cdot \left(m(\phi) \nabla \mu \right) \label{phase_modelI_simp}  \\
& \mu = \dfrac{\sigma}{\varepsilon} \frac{df}{d\phi} - \sigma \varepsilon \Delta \phi  + \alpha p \label{chem_pot_modelI_simp}  \\
 &\rho \dot{\mathbf{v}}= - \nabla p -\phi \nabla (\mu-\alpha p) + \nabla \cdot \left( \eta(\phi) \left( 2  \mathbf{D}+ \lambda (\nabla \cdot \mathbf{v}) \mathbf{I} \right) \right) - \hat{\rho} g  {\boldsymbol{\jmath}}   \label{momentum_modelI_simp}  \\
 & \nabla \cdot \mathbf{v} = \alpha \nabla \cdot \left(m(\phi) \nabla \mu \right) \label{div_v_modelI_simp} \\ 
 & \dot{\rho} + \rho \nabla \cdot  \mathbf{v}= 0  \label{density_alg_modelI_simp_a} \\ 
& \hat{\rho} = \rho_1 \frac{1+\phi}{2} + \rho_2 \frac{1-\phi}{2}.  \label{density_alg_modelI_simp_b} 
\end{empheq}
\end{subequations}
It can be observed, however, that the models~\eqref{final_NSCH} and~\eqref{modelI_simp} are \emph{equivalent}, by redefining the pressure! Using the relation \eqref{relation2}, if the pressure $p$ in \eqref{chem_pot_modelI_simp} and~\eqref{momentum_modelI_simp} is redefined as $p = \frac{\tilde{p}}{1-\alpha \phi}$, then one obtains \eqref{chem_pot_modelII_simp} and~\eqref{momentum_modelII_simp}.
\par
\begin{remark}[Equivalence with Shen, Yang \&~Wang]
The form obtained in~\eqref{modelI_simp} is equivalent to the model by Shen, Yang \&~Wang\cite{SheYanWanCCP2013}, which was derived using different arguments, not invoking the Coleman--Noll procedure. Indeed, their Eqs.~(2.13a)--(2.13d) on their page~1050, can be obtained from our~\eqref{phase_modelI_simp}--\eqref{div_v_modelI_simp} upon simply redefining our $\rho_2 = 2\tilde{\rho}_2 - \rho_1$, to account for the fact that their phase variable ranges from~$0$ to~$1$ (as opposed to~${-}1$ to~$1$ in our case). With that substitution, the $\alpha$ in~\eqref{chem_pot_modelI_simp}--\eqref{div_v_modelI_simp} changes to~$\frac{\tilde{\rho}_2-\rho_1}{\tilde{\rho}_2}$, and $\hat{\rho}(\phi) = \rho_1\phi + \tilde{
\rho}_2(1-\phi)$, which correspond to the model of Shen~\emph{et al.}
%
\end{remark}
\begin{remark}[Equivalence with Aki, Dreyer, Giesselmann~\& Kraus]
The form obtained in~\eqref{modelI_simp} is also equivalent to the isothermal form (and without phase transition) of the generalized model by Aki, Dreyer, Giesselmann~\& Kraus\cite{AkiDreyM3AS2014}. Indeed, using~\eqref{simplify} and the identity 
\begin{alignat*}{2}
  \nabla\cdot \Big( 
  \nabla \phi \otimes \nabla \phi - (\phi \Delta \phi + \tfrac{1}{2}|\nabla \phi|^2) \mathbf{I}
  \Big) = -\phi\nabla \Delta \phi\,,
\end{alignat*}
Eq.~\eqref{momentum_modelI_simp} can be written as
\begin{alignat}{2}
\label{akiMom}
  \rho \dot{\mathbf{v}}= - \nabla p + \sigma \varepsilon  \phi \nabla \Delta \phi 
  + \nabla P(\phi)
  + \nabla \cdot \left( \eta(\phi) \left( 2  \mathbf{D}+ \lambda (\nabla \cdot \mathbf{v}) \mathbf{I} \right) \right) - \hat{\rho} g  {\boldsymbol{\jmath}}
\end{alignat}
where $P(\phi) :=  \phi \frac{dW}{d\phi}(\phi) - W(\phi)$ is the thermodynamic pressure as used by Aki~\emph{et al.}, and $W(\phi) := \sigma \varepsilon^{-1} f(\phi)$. The system given by~\eqref{phase_modelI_simp}, \eqref{chem_pot_modelI_simp}, \eqref{akiMom}, \eqref{div_v_modelI_simp} is now exactly equal to the one in Aki~\emph{et al.} on page~828 (cf.~ their Eqs.~(2.16)--(2.20)), upon setting the mobilities in their model to $m_j = \frac{m(\phi)}{c_+^2}$ and $m_r=0$.
\end{remark}
These equivalences unify the seemingly different quasi-incompressible NSCH models based on the mass-averaged velocity and volume-fraction phase variable, which have all been derived in different ways.
%
\subsection{Non-dimensionalization} 
\label{sec:nondim}
We perform the non-dimensionalization of \eqref{final_NSCH} based on physical properties of the
liquid associated with the density $\rho_1$, which are the characteristic scales of length, $L_{\ast}$, velocity, $U_{\ast}$ and
$P_{\ast}=\sigma/L_{\ast}$. Considering that $\phi$ is already a dimensionless phase field
variable, the governing dimensionless variables are given by:
\begin{align*}
  \bar{\mathbf{x}} &= \frac{\mathbf{x}}{L_{\ast}} &
  \bar{\mathbf{v}} &= \frac{\mathbf{v}}{U_{\ast}} &
  \bar{t} &= t \, \frac{L_{\ast}}{U_{\ast}} &
  \bar{p} &= \frac{p}{P_{\ast}} &
  \bar{\mu} &= \frac{\mu}{P_{\ast}} & \bar{\rho}&=\frac{\rho}{\rho_1}
\end{align*}
where $P_{\ast}$ is used for both dimensionless pressure and chemical potential. Suppressing the over bar for the sake of simplicity, 
the dimensionless quasi-incompressible NSCH system \eqref{final_NSCH} writes as:

\begin{subequations} \label{NSCH_nondim}
\begin{empheq}[box=\widefbox]{align}
& \dot{\phi} + \phi  \nabla \cdot \mathbf{v}  = \frac{1}{\mathrm{Pe}} \triangle \mu   \label{phase_modelII_simp_nondim} \\
&\mu = \frac{1}{\mathrm{Cn}} \frac{df}{d\phi} - \mathrm{Cn} \Delta \phi  - \frac{p}{\hat{\rho}} \frac{d \hat{\rho}}{d \phi}  \label{chem_pot_modelII_simp_nondim}  \\
&\rho \dot{\mathbf{v}} = - \frac{1}{\mathrm{We}} \left( \nabla p + \phi \nabla \mu  - \frac{p}{\hat{\rho}} \frac{d \hat{\rho}}{d \phi}  \nabla \phi \right) \nonumber \\
  & \hspace{2cm} +\frac{1 }{\mathrm{Re}} \nabla \cdot \left( 2  \mathbf{D}+ \lambda (\nabla \cdot \mathbf{v}) \mathbf{I} \right) - \frac{1}{\mathrm{Fr}^{2}} \hat{\rho}{\boldsymbol{\jmath}} \label{momentum_modelII_simp_nondim} \\
 & \nabla \cdot \mathbf{v} = \frac{\alpha}{\mathrm{Pe}} \triangle \mu  \label{div_v_modelII_simp_nondim}\\
 &  \dot{\rho} + \rho \nabla \cdot  \mathbf{v}= 0  \label{density_alg_modelII_simp_nondim_a} \\
 & \hat{\rho} = \frac{1+\phi}{2} + \frac{\rho_2}{\rho_1} \frac{1-\phi}{2},  \label{density_alg_modelII_simp_nondim_b}
\end{empheq}
\end{subequations}
where $\mathrm{Pe} = U_{\ast} L_{\ast}^2/m\sigma$ is the diffusional $P\acute{e}clet$  number, $\mathrm{Cn} = \varepsilon/L_{\ast}$ is the Cahn number as a measure of the interface thickness, $\mathrm{We} = \rho_1 U_{\ast
}^2 L_{\ast}/\sigma$ is the Weber number, $\mathrm{Re} = \rho_1 U_{\ast} L_{\ast}/\eta$ is the Reynolds number and $\mathrm{Fr} = U_{\ast}/\sqrt{g L_{\ast}}$ is the Froude number, which expresses the relative strengths of the inertial and the gravitational forces. Additionally, $\alpha $ is the ratio of specific densities as already defined in \eqref{alpha} and the dimensionless $d \hat{\rho} / d \phi$ writes:
\begin{equation}
\frac{d \hat{\rho}}{d \phi} = \frac{1}{2} \left( 1+ \frac{\rho_2}{\rho_1} \right). \nonumber
\end{equation}
The dimensionless form of \eqref{modelI_simp} can be obtained similarly.
\par
In addition to the main system, the dimensionless energy is obtained from \eqref{energy} by
rescaling $\bar{E} = E/\rho_{1} U_{\ast}^2$ as
\begin{equation}
E \left(\phi, \rho, \hat{\rho}, \mathbf{v} \right) = \int_{\Omega} \left( \frac{1}{2} \rho |\mathbf{v}|^2 + \frac{1}{\mathrm{We} \, \mathrm{Cn}} f(\phi) + \frac{\mathrm{Cn}}{2 \, \mathrm{We}} |\nabla \phi|^2 + \frac{1}{\mathrm{Fr}^2}  \hat{\rho} y\right) \, d\Omega. \label{energy_nondim}
\end{equation}

\subsection{Interface profile}
\label{sec:ffinterface}
To elucidate the manner in which the fluid--fluid interface is represented in the 
quasi-incompressible NSCH model~\eqref{NSCH_nondim}, we consider the particular case of a 
stationary planar fluid--fluid interface in the absence of gravity. Under suitable boundary conditions, it
follows from~\eqref{NSCH_nondim} that $\mathbf{v}=0$ and $\mu=\mathrm{const}$.
Denoting by~$s$ a coordinate normal to the 
interface and centered at the interface, equations~\eqref{NSCH_nondim} then imply
\begin{equation}
\label{eq:NSCH_stat}
\left.
\begin{aligned}
\frac{d{}p}{ds}-\frac{p}{\hat{\rho}}\frac{d\hat{\rho}}{d\phi}\frac{d\phi}{ds}&=0
\\
\mu
-\frac{1}{\mathrm{Cn}}\frac{df}{d\phi}
+\mathrm{Cn}\frac{d^2\phi}{ds^2}
+\frac{p}{\hat{\rho}}\frac{d\hat{\rho}}{d\phi}&=0
\end{aligned}
\right\}\text{ in }\mathbb{R}
\end{equation}
We insist that the mixture reduces to the pure species and exhibits a uniform pressure 
(i.e. trace of the stress) away from the interface:
\begin{equation}
\label{eq:farfield}
\lim_{s\to\pm\infty}\phi(s)=\pm{}1\qquad
\lim_{s\to\pm\infty}dp(s)-\tr\boldsymbol{\zeta}(s)=dp_{\infty}
\end{equation}
where $\boldsymbol{\zeta}$ represents the dimensionless capillary-stress tensor according to
\begin{equation}
\label{eq:zeta_spec}
\begin{aligned}
\boldsymbol{\zeta}&=-\mathrm{Cn}\nabla\phi\otimes\nabla\phi+\mathbf{I}\bigg(\frac{\mathrm{Cn}}{2}\big|\nabla\phi\big|^2
+\frac{1}{\mathrm{Cn}}f(\phi)-\mu\phi\bigg)
\\
&=-\frac{\mathrm{Cn}}{2}\bigg|\frac{d\phi}{ds}\bigg|^2
+\frac{1}{\mathrm{Cn}}f(\phi)-\mu\phi
\end{aligned}
\end{equation}
and $\operatorname{tr}\boldsymbol{\zeta}$ represents its trace. The second identity in~\eqref{eq:zeta_spec}
holds in the one-dimensional case under consideration.
The first equation in~\eqref{eq:NSCH_stat} can be solved via separation of variables to obtain the general 
solution~$p=C\rho(\phi)$ for some constant~$C$. It can then be verified by substitution that
\begin{equation}
\label{eq:sol_ODE}
\phi(s)=\tanh\bigg(\frac{s}{\sqrt{2}\,\mathrm{Cn}}\bigg)
\qquad
p(s)=p_{\infty}+p_{\infty}\bigg(\frac{\rho_1-\rho_2}{\rho_1+\rho_2}\bigg)\phi(s)
\qquad
\mu=-p_{\infty}\bigg(\frac{\rho_1-\rho_2}{\rho_1+\rho_2}\bigg)
\end{equation}
solves~\eqref{eq:NSCH_stat}--\eqref{eq:farfield}. It is to be noted that, in particular,  
the order parameter $\phi$ assumes the typical tanh\nobreakdash-form.\cite{Jacqmin:1999fk,Jacqmin:2000kx}

\begin{remark}
The fluid--fluid surface tension can be conceived of as the increase in the free energy that accompanies an increase in surface area of the fluid-fluid meniscus.\cite{Gennes:2004ai} Accordingly, the surface tension associated with the 
solution~\eqref{eq:sol_ODE} of the order parameter can be derived from Eq.~\eqref{energy_nondim}~as:
\begin{equation}
\label{eq:Sigma}
\begin{aligned}
\frac{1}{\mathrm{We}}
\int_{-\infty}^{+\infty} 
\bigg(\frac{\mathrm{Cn}}{2}\bigg|\frac{d\phi}{ds}\bigg|^2+\frac{1}{\mathrm{Cn}}f(\phi)\bigg)\,ds
&=
\frac{1}{\mathrm{We}}
\int_{-\infty}^{\infty}\bigg(\frac{1}{2\,\mathrm{Cn}}\bigg)\operatorname{sech}^4\bigg(\frac{s}{\sqrt{2}\,\mathrm{Cn}}\bigg)\,ds
\\
&=
\frac{2\sqrt{2}}{3\,\mathrm{We}}
\end{aligned}
\end{equation}
with $\operatorname{sech}(\cdot)$ the hyperbolic-secant function.
\end{remark}

\begin{remark}
The energy density in~\eqref{eq:Sigma} is a strictly positive function that is essentially localized in the
vicinity of the interface, and that collapses onto the interface in the sharp-interface limit $\mathrm{Cn}\to+0$;
see Figure~\ref{fig:sech4}.
\end{remark}
\begin{figure}
\begin{center}
\includegraphics[width=0.75\textwidth]{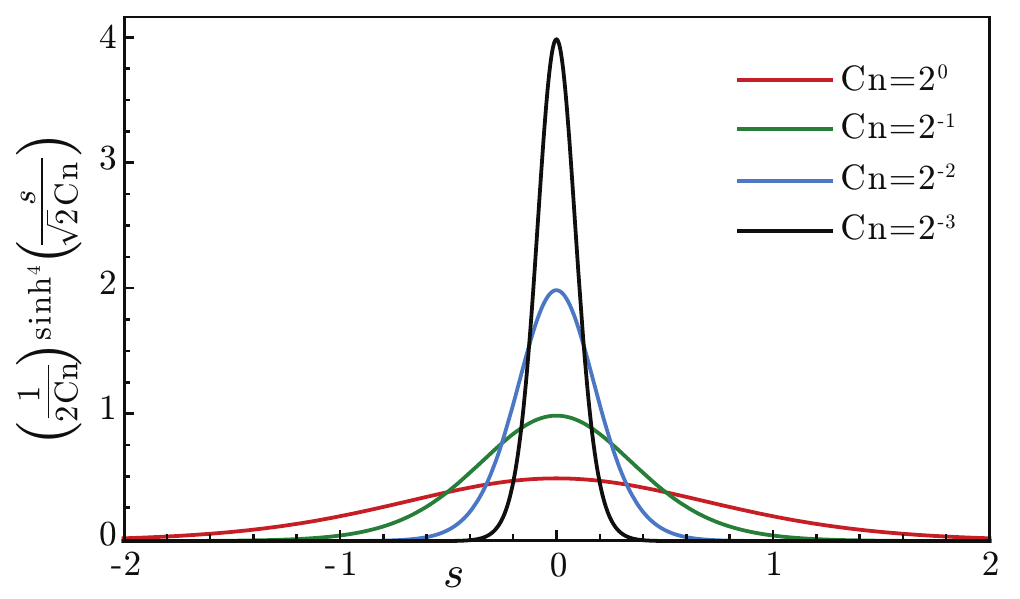}
\end{center}
\caption{Illustration of the energy density in~\eqref{eq:Sigma} for $\mathrm{Cn}\in{}2^{\{-3,-2,-1,0\}}$.
\label{fig:sech4}}
\end{figure}

%
\section{Energy-Dissipative Numerical Method}
\label{section3}
%
We constructed our model \eqref{NSCH_nondim} based on thermodynamic consistency, which is equivalent to satisfying the second axiom of thermodynamics. The axiom demands entropy production which in turn implies dissipation of total energy $E$, i.e. according to \eqref{energy_nondim}, $\displaystyle d_{t}E \leq 0$. To maintain stability in numerical computations, this physical dissipation property should be observed at the discrete level, as it is in the continuous PDE model. In this section, for completeness, we first present the continuous energy dissipation for the weak formulation of \eqref{NSCH_nondim}. Then we present a linearly implicit, first order accurate time-integration scheme. We restrict ourselves here only to time-discretization and prove the discrete dissipation for a space-continuous semi-discrete model.
%
\subsection{Energy dissipative weak formulation} 
We assume that the natural boundary conditions \eqref{boun_con} hold. The system of differential equations in \eqref{NSCH_nondim} subject to \eqref{boun_con} can be condensed into the following weak form:
Find $\phi, \mu \in {H^{1}}(\Omega)$, $\mathbf{v} \in {\mathbf{H}}^{1}(\Omega)$, $p \in {L^{2}}(\Omega)$ such that
\begin{subequations}
\label{NSCH_nondim_weak}
\begin{align}
& \int_{\Omega} \Big(\dot{\phi} \omega+ \phi  \nabla \cdot \mathbf{v} \omega  +  \frac{1}{\mathrm{Pe}} \nabla \mu \cdot \nabla \omega \Big) \,d \Omega = 0 , \quad \forall \omega \in  H^{1}(\Omega) \label{phase_modelII_simp_nondim_weak} \\
& \int_{\Omega} \left( \mu \psi  - \frac{1}{\mathrm{Cn}} \frac{df}{d\phi} \psi - \mathrm{Cn} \nabla \phi \cdot  \nabla \psi +\frac{p}{\hat{\rho}} \frac{d \hat{\rho}}{d \phi}  \psi \right)  \,d \Omega = 0,  \quad \forall \psi \in H^{1}(\Omega) \label{chem_pot_modelII_simp_nondim_weak}  \\
 & \int_{\Omega} \left( \rho \dot{\mathbf{v}} \cdot \boldsymbol{\chi} + \dfrac{1}{\mathrm{We}} (  - p \, (\nabla \cdot \boldsymbol{\chi} )+ \phi \nabla \mu \cdot \boldsymbol{\chi}  - \frac{p}{\hat{\rho}} \frac{d \hat{\rho}}{d \phi}  \nabla \phi  \cdot  \boldsymbol{\chi} )  +\frac{2}{\mathrm{Re}}  \mathbf{D}  : \nabla  \boldsymbol{\chi} \right. \nonumber \\
 &\hspace{1.5cm} \left. + \frac{\lambda}{\mathrm{Re}}  (\nabla \cdot \mathbf{v} )\, (\nabla \cdot \boldsymbol{\chi})  +  \frac{1}{\mathrm{Fr}^{2}}\hat{\rho}{\boldsymbol{\jmath}}  \cdot  \boldsymbol{\chi} \right)  \, d \Omega = 0, \quad \forall \boldsymbol{\chi} \in {\mathbf{H}}^{1}(\Omega) \label{momentum_modelII_simp_nondim_weak} \\
 & \int_{\Omega} \left( (\nabla \cdot \mathbf{v}) \, \theta  + \frac{\alpha}{\mathrm{Pe}} \nabla \mu \cdot \nabla \theta \right) \,d \Omega = 0,  \quad \forall \theta \in  L^{2}(\Omega) \label{div_v_modelII_simp_nondim_weak}
\end{align}
\end{subequations}
for a.e. $t\in (0,T)$. 

The specification of the function spaces in the weak formulation~\eqref{NSCH_nondim_weak}
is formal and a consideration of existence and stability is beyond the scope of this work. However, additional
(a-posteriori) conditions are required, e.g. $\phi\in{}L^{\infty}(\Omega)$ and $\phi\in[-1,1]$ a.e. in~$\Omega$,
to ensure that the integrals in~\eqref{NSCH_nondim_weak} are bounded. The mixture densities $\hat{\rho}$ and $\rho$ 
have not been included in~\eqref{NSCH_nondim_weak}, because~$\hat{\rho}$ is defined via the 
algebraic relation~\eqref{density_alg_modelII_simp_nondim_b} and it can be replaced with its definition in \eqref{NSCH_nondim_weak}. Similarly, $\rho$ is related to the model with respect to the mass conservation equation \eqref{density_alg_modelII_simp_nondim_a} and it is not considered as part of the system \eqref{NSCH_nondim_weak}. 
However, to ensure that the integrals in the weak form~\eqref{NSCH_nondim_weak} are appropriately bounded, we require the mixture densities to satisfy $\rho, \hat{\rho} \in {L^{\infty}}(\Omega,\mathbb{R}_{>0})$. 
%
\begin{theorem} \label{thm_1}
Let $\phi, \mu, \mathbf{v}, p$ be a sufficiently smooth solution to \eqref{NSCH_nondim_weak} subject to the boundary condition  \eqref{boun_con}. Assume also that $\rho$ and $\hat{\rho}$ according to~\eqref{density_alg_modelII_simp_nondim_a} and \eqref{density_alg_modelII_simp_nondim_b}, respectively, are positive. Then the following energy dissipation relation holds:
\begin{equation}
\label{energy_dissipation}
\begin{aligned}
&\frac{d }{dt} \int_{\Omega} \left( \frac{1}{2}\rho |\mathbf{v}|^2 + \frac{1}{\mathrm{We} \, \mathrm{Cn}}   f(\phi) + \frac{\mathrm{Cn}}{ 2 \, \mathrm{We}} |\nabla \phi|^2  +  \frac{1}{\mathrm{Fr}^2} \hat{\rho} y \right) \, d \Omega  \\
& \hspace{0.6cm} = -\frac{1}{\mathrm{Re}}  \|\nabla  \boldsymbol{v} \|_{L^{2}}^{2} 
- \frac{1+\lambda}{\mathrm{Re}} \| \nabla \cdot \mathbf{v} \|_{L^{2}}^2 - \frac{1}{\mathrm{Pe}\, \mathrm{We}} \| \nabla \mu \|_{L^{2}}^2 \leq 0.
\end{aligned}
\end{equation}
\end{theorem}
\begin{proof}
Setting $\mathbf{\boldsymbol{\chi}} = \mathbf{v}$ for the momentum equation \eqref{momentum_modelII_simp_nondim_weak} and invoking integration by parts and the boundary conditions~\eqref{boun_con}, we obtain
\begin{equation}
\label{momentum_modelII_simp_nondim_weak_energy}
\begin{aligned}
 & \frac{d }{d t} \int_{\Omega}  \frac{1}{2}\rho |\mathbf{v}|^2 \,d \Omega +\frac{1}{\mathrm{Re}}  \|\nabla  \boldsymbol{v} \|_{L^{2}}^{2} + \frac{1+\lambda}{\mathrm{Re}} \| \nabla \cdot \mathbf{v} \|_{L^{2}}^2 + \int_{\Omega} \frac{1}{\mathrm{Fr}^2} \hat{\rho} {\boldsymbol{\jmath}}  \cdot  \mathbf{v}  \, d \Omega  \\
& + \int_{\Omega} \frac{1}{\mathrm{We}} 
 \left(  - p (\nabla \cdot \mathbf{v} )+ \phi \nabla \mu \cdot \mathbf{v}  - \frac{p}{\hat{\rho}} \frac{d \hat{\rho}}{d \phi}  \nabla \phi \cdot \mathbf{v}\right) \, d \Omega = 0. 
\end{aligned}
\end{equation}
In \eqref{momentum_modelII_simp_nondim_weak_energy} we have used the following identity, which is obtained by the mass conservation equation \eqref{density_alg_modelII_simp_nondim_a}, integration by parts and the homogeneous boundary condition $\mathbf{v}|_{\partial \Omega} = 0$:
\begin{equation}
\int_{\Omega} \rho (\mathbf{v} \cdot \nabla) \mathbf{v} \cdot \mathbf{v} \, d \Omega = - \int_{\Omega} \frac{1}{2} \nabla \cdot (\rho \mathbf{v}) |\mathbf{v}|^{2}  \, d \Omega =\int_{\Omega} \frac{1}{2} \partial_t\rho |\mathbf{v}|^{2}  \, d \Omega. \nonumber
\end{equation}
Then, set $\omega = \frac{1}{\mathrm{We}}(\mu + \frac{p}{\hat{\rho}} \frac{d \hat{\rho}}{ d \phi} ) + \frac{1}{\mathrm{Fr}^2} \frac{d \hat{\rho}}{ d \phi}  y $ and $\displaystyle  \psi = - \partial_t \phi / \mathrm{We}$ for the phase 
equation~\eqref{phase_modelII_simp_nondim_weak} and the chemical potential equation \eqref{chem_pot_modelII_simp_nondim_weak}, respectively, to obtain
\begin{align}
& \int_{\Omega}  \frac{1}{\mathrm{We}}\left( \partial_t \phi +  \nabla \cdot(\phi \mathbf{v}) \right) 
\left(\mu + \frac{p}{\hat{\rho}} \frac{d \hat{\rho}}{ d \phi} \right) \,d \Omega +  \frac{1}{\mathrm{Pe} \, \mathrm{We}} \| \nabla \mu \|_{L^{2}}^2 \nonumber \\
 &+ \int_{\Omega} \frac{1}{\mathrm{Pe} \,\mathrm{ We}} \nabla \mu \cdot  \nabla \left( \frac{p}{\hat{\rho}} \frac{d{\rho}}{ d \phi} \right)  \,d \Omega   +\int_{\Omega} \left( \partial_t \phi +  \nabla \cdot(\phi \mathbf{v}) \right) \frac{1}{\mathrm{Fr}^2} \frac{d \hat{\rho}}{ d \phi}  y  \,d \Omega  \nonumber \\
 &+ \int_{\Omega}  \frac{1}{\mathrm{Pe}} \nabla \mu \cdot \nabla \left( \frac{1}{\mathrm{Fr}^2} \frac{d \hat{\rho}}{ d \phi}  y\right)  \,d \Omega  = 0, \label{phase_modelII_simp_nondim_weak_energy}  \\
& \int_{\Omega} \frac{1}{\mathrm{We}}\left( - \mu \partial_t \phi  + \frac{1}{\mathrm{Cn}} \frac{df}{d\phi} \partial_t \phi  + \mathrm{Cn} \nabla \phi \cdot  \partial_{t}(\nabla \phi) - \frac{p}{\hat{\rho}} \frac{d \hat{\rho}}{d \phi}  \, \partial_t \phi  \right)  \,d \Omega = 0,  \label{chem_pot_modelII_simp_nondim_weak_energy} 
\end{align}
Similarly, the choice $ \theta = - \frac{1}{\alpha} \left( \frac{1}{\mathrm{We}} \frac{p}{\hat{\rho}} \frac{d \hat{\rho}}{ d \phi}  + \frac{1}{\mathrm{Fr}^2} \frac{d \hat{\rho}}{ d \phi} y \right)$ for \eqref{div_v_modelII_simp_nondim_weak} gives
\begin{equation}
\label{div_v_modelII_simp_nondim_weak_energy}
\begin{aligned}
\int_{\Omega} & - \frac{1}{\alpha}\nabla \cdot  \mathbf{v}\left(  \frac{1}{\mathrm{We}} \frac{p}{\hat{\rho}} \frac{d \hat{\rho}}{ d \phi}  + \frac{1}{\mathrm{Fr}^2} \frac{d \hat{\rho}}{ d \phi} y\right) \,d \Omega  \\
&- \int_{\Omega} \frac{1}{\mathrm{Pe}} \nabla \mu \cdot \nabla \left(   \frac{1}{\mathrm{We}} \frac{p}{\hat{\rho}} \frac{d \hat{\rho}}{d \phi}  + \frac{1}{\mathrm{Fr}^2} \frac{d \hat{\rho}}{ d \phi} y \right)  = 0.  
\end{aligned}
\end{equation}
Adding equations \eqref{momentum_modelII_simp_nondim_weak_energy}--\eqref{div_v_modelII_simp_nondim_weak_energy} results in 
\begin{align}
&\frac{d }{dt} \int_{\Omega} \left( \frac{1}{2}\rho |\mathbf{v}|^2 + \frac{1}{\mathrm{We} \, \mathrm{Cn}}   f(\phi) + \frac{\mathrm{Cn}}{ 2 \, \mathrm{We}} |\nabla \phi|^2  +  \frac{1}{\mathrm{Fr}^2} \hat{\rho} y \right) \, d \Omega \nonumber \\
& =  -\frac{1}{\mathrm{Re}}  \|\nabla  \boldsymbol{v} \|_{L^{2}}^{2} - \frac{1+\lambda}{\mathrm{Re}} \| \nabla \cdot \mathbf{v} \|_{L^{2}}^2 - \frac{1}{\mathrm{Pe} \, \mathrm{We}} \| \nabla \mu \|_{L^{2}}^2 \nonumber \\
& + \int_{\Omega}  \frac{1}{\mathrm{We}} \left(  p (\nabla \cdot \mathbf{v}) + \frac{p}{\hat{\rho}} \frac{d \hat{\rho}}{d \phi} \left (
\nabla \phi \cdot \mathbf{v} -  \nabla \cdot (\phi \mathbf{v}) + \frac{1}{\alpha} \nabla \cdot \mathbf{v}\right)\right) \, d \Omega \nonumber \\
& + \int_{\Omega} \frac{1}{\mathrm{Fr}^{2}} \left( -\hat{\rho} {\boldsymbol{\jmath}} \cdot \mathbf{v} - \frac{d \hat{\rho}}{ d \phi} y \left( \nabla \cdot (\phi \mathbf{v}) - \frac{1}{\alpha} \nabla \cdot \mathbf{v} \right)  \right) \, d \Omega. \nonumber
\end{align}
Using integration by parts, the relation \eqref{relation2} and the fact that $\hat{\rho}$ is a function of $\phi$ as defined in \eqref{density_alg_modelII_simp_nondim_b}, one can deduce that the last two integrals vanish:
\begin{equation}
\label{vanish_integrals}
\begin{aligned}
& \int_{\Omega}  \frac{1}{\mathrm{We}} \left(  p (\nabla \cdot \mathbf{v}) + \frac{p}{\hat{\rho}} \frac{d \hat{\rho}}{d \phi} \left (
\nabla \phi \cdot \mathbf{v} -  \nabla \cdot (\phi \mathbf{v}) + \frac{1}{\alpha} \nabla \cdot \mathbf{v}\right)\right) \, d \Omega = 0 \\
& \int_{\Omega} \frac{1}{\mathrm{Fr}^{2}} \left( -\hat{\rho} {\boldsymbol{\jmath}} \cdot \mathbf{v} - \frac{d \hat{\rho}}{ d \phi} y \left( \nabla \cdot (\phi \mathbf{v}) - \frac{1}{\alpha} \nabla \cdot \mathbf{v} \right)  \right) \, d \Omega = 0, 
\end{aligned}
\end{equation}  
which completes the proof.
\end{proof}
The energy-dissipation relation in Theorem~\ref{thm_1} is a fundamental structural property.
Noting that the energy in~\eqref{energy_nondim} corresponds to a convex functional in~$\mathbf{v}$ and~$\phi$,
the energy-dissipation relation endows the quasi-incompressible NSCH system~\eqref{NSCH_nondim_weak} with 
stability in the Lyapunov sense. The energy-dissipation property should be preserved in discrete 
approximations: see Section~\ref{Section3_2:en_dis_time}. Additionally, conservation of mass and phase are other structural properties of the NSCH model to be retained in discrete approximations. The continuity equation 
\eqref{density_alg_modelII_simp_nondim_a} and the boundary conditions in \eqref{boun_con} imply conservation of mass:
\begin{align}
\frac{d}{dt}\int_{\Omega} \rho \, d \Omega &= \int_{\Omega} \partial_{t} \rho  \, d\Omega =  \int_{\Omega} - \nabla \cdot (\rho \mathbf{v})  \, d\Omega =  \int_{\partial \Omega} - \rho \mathbf{v} \cdot \mathbf{n} \, dS = 0  \nonumber 
\end{align}
Similarly, conservation of phase $\phi$ follows from the expression for $\hat{\rho}$ as a function of $\phi$, equations \eqref{phase_modelII_simp_nondim} and \eqref{div_v_modelII_simp_nondim}, relation \eqref{relation2} and the boundary conditions:
\begin{equation}
\label{conservation_of _phase}
\begin{aligned}
\frac{d}{dt}\int_{\Omega} \hat{\rho} \, d\Omega &= \int_{\Omega}  \frac{d \hat{\rho}}{d \phi} \frac{\partial\phi}{\partial t} \, d\Omega = \int_{\Omega} \frac{d \hat{\rho}}{d \phi}  \left( - \nabla \cdot (\phi \mathbf{v}) + \frac{1}{\mathrm{Pe}} \triangle \mu  \right) \, d \Omega \\
 &  =\int_{\Omega} \frac{d \hat{\rho}}{d \phi}  \left( - \nabla \cdot (\phi \mathbf{v}) + \frac{1}{\alpha}  \nabla \cdot  \mathbf{v} \right) \, d \Omega = \int_{\Omega} - \nabla \cdot (\hat{\rho} \mathbf{v})  \, d\Omega  \\
 &=  \int_{\partial \Omega} - \hat{\rho} \mathbf{v} \cdot \mathbf{n} \, dS = 0.  
\end{aligned}
\end{equation}
Because $ {d \hat{\rho}}/{d \phi}$ is constant, the chain of identities in \eqref{conservation_of _phase} implies 
\begin{equation}
\int_{\Omega}\frac{\partial\phi}{\partial t}\,d\Omega
=
\frac{d}{dt}\int_{\Omega}\phi\,d\Omega=0
\end{equation}
That is, the phase $\phi$ is conserved.
%
\subsection{Linear energy-stable time-integration scheme} 
\label{Section3_2:en_dis_time}
In this section, we introduce a linear finite difference time-integration scheme for model \eqref{NSCH_nondim}. Instead 
of the weak form, we propose the scheme in the context of the strong form of the PDE to present it with a clear algorithm chart.
\par
We consider a partitioned of the time interval $(0,T)$ into $N$ sub-intervals of constant time-step size, 
$ \Delta t = t^{n+1} - t^{n}$ for $n = 0, 1, 2, \ldots, N-1$. Algorithm~\ref{algorithm_scheme} presents a coupled, 
first-order accurate and energy-dissipative time-integration scheme for the NSCH system~\eqref{NSCH_nondim}. The 
properties of the scheme are stated and proved in Theorem~\ref{thm2}.
\begin{newalgorithm}
\begin{framed}
Given $\phi^{0}, {\mathbf{v}}^{0}$. Initialize $\rho^{0}$ using the algebraic definition by
$$\rho^{0} = \frac{1+\phi^{0}}{2} + \frac{\rho_2}{\rho_1} \frac{1-\phi^{0}}{2}.$$ 
For $n = 0, 1, 2, \ldots, N-1$,
\begin{itemize}
\item[] \textit{Step 1.} Compute
 \begin{equation}
\hat{\rho}^{n}= \frac{1+\phi^{n}}{2} + \frac{\rho_2}{\rho_1} \frac{1-\phi^{n}}{2}. \label{density_alg_modelII_simp_nondim_timedis}
 \end{equation}
 \end{itemize}
 \begin{itemize}
\item[] \textit{Step 2.} Solve the following system to obtain $\phi^{n+1},\mu^{n+1}, {\mathbf{v}}^{n+1} $ and $p^{n+1}$
\begin{subequations}
\begin{align}
& \hspace{-0.5cm} \frac{\phi^{n+1}-\phi^{n}}{\Delta t}+  \nabla \cdot  (\phi^{n} \mathbf{v} ^{n+1}) = \frac{1}{\mathrm{Pe}} \triangle \mu^{n+1}   \label{phase_modelII_simp_nondim_timedis} \\
& \hspace{-0.5cm} \mu^{n+1} =\frac{\beta}{\mathrm{Cn}} (\phi^{n+1} - \phi^{n}) +  \frac{1}{\mathrm{Cn}} f'(\phi^{n}) - \mathrm{Cn} \Delta \phi^{n+1}  - \frac{1}{\hat{\rho}^{n}} \frac{d \hat{\rho}^{n}}{d \phi^{n}} p^{n+1} \label{chem_pot_modelII_simp_nondim_timedis}  \\
 & \hspace{-0.5cm}\rho^{n} \frac{\mathbf{v}^{n+1} - \mathbf{v}^{n}}{\Delta t} + \rho^{n} \mathbf{v}^{n} \cdot \nabla \mathbf{v}^{n+1}= \nonumber \\
  & \hspace{2.0cm} - \frac{1}{\mathrm{We}}\left( \nabla p^{n+1} + \phi^{n} \nabla \mu^{n+1}  - \frac{1}{\hat{\rho}^{n}} \frac{d \hat{\rho}^{n}}{d \phi^{n}} p^{n+1} \nabla \phi^{n} \right) \nonumber \\
& \hspace{2.0cm} +\frac{1}{\mathrm{Re}} \nabla \cdot\left( 2  \mathbf{D}^{n+1}+ \lambda (\nabla \cdot \mathbf{v}^{n+1}) \mathbf{I} \right) - \frac{1}{\mathrm{Fr}^{2}} \hat{\rho}^{n} {\boldsymbol{\jmath}} \label{momentum_modelII_simp_nondim_timedis} \\
 &\hspace{-0.5cm} \nabla \cdot \mathbf{v}^{n+1} = \frac{\alpha}{\mathrm{Pe}} \triangle \mu^{n+1}  \label{div_v_modelII_simp_nondim_timedis}
\end{align}
\end{subequations}
with boundary conditions 
\begin{equation}
\nabla \phi^{n+1} \cdot \mathbf{n} = \nabla \mu^{n+1} \cdot \mathbf{n} = 0, \quad \mathbf{v}^{n+1} = \mathbf{0} \quad \text{on}\, \partial \Omega. \label{boun_con_disc}
\end{equation}
 \end{itemize}
  \begin{itemize}
\item[] \textit{Step 3.} Update $\rho^{n}$ to $\rho^{n+1}$  for $n \geq 1$, using the mass conservation equation:
\begin{equation}
\rho^{n+1}  = \rho^{n} -  \Delta t \nabla \cdot (\rho^{n} \mathbf{v}^{n}). \label{density_mass_modelII_simp_nondim_timedis}
 \end{equation}
  \end{itemize}
 This completes one time step, update $n$ to $n+1$.
\end{framed}
\caption{Energy-dissipative linearly-implicit time-stepping scheme}
\label{algorithm_scheme}
\end{newalgorithm} 
Except for the initialization step in which we take $\hat{\rho}^{0} = \rho^{0}$, each time step in the time-integration algorithm involves both densities $\hat{\rho}^{n}$ and $\rho^{n}$ according to~\eqref{density_alg_modelII_simp_nondim_timedis} and \eqref{density_mass_modelII_simp_nondim_timedis}, respectively. By virtue of this approach and by delaying the 
transport velocity in the nonlinear convective term~\eqref{momentum_modelII_simp_nondim_timedis} to~$t^n$, 
Algorithm~\ref{algorithm_scheme} is \emph{linearly implicit}, i.e. equations \eqref{density_alg_modelII_simp_nondim_timedis}--\eqref{density_mass_modelII_simp_nondim_timedis} are linear in $\phi^{n+1}, \mu^{n+1}, \mathbf{v}^{n+1}$ and~$p^{n+1}$.
\par
It is to be noted that a stabilizing term $\frac{\beta}{\mathrm{Cn}} (\phi^{n+1} - \phi^{n})$ has been introduced in the definition of the chemical potential in~\eqref{chem_pot_modelII_simp_nondim_timedis}. The constant stabilization 
parameter $\beta$ depends on the choice of~$f(\phi)$. Specifically, for the $f(\phi)$ function in~\eqref{f_phi}, 
the scheme becomes stable if the stabilization constant satisfies~$\beta\geq{}1$.

\begin{theorem} \label{thm2}
Consider the time-integration scheme in Algorithm~\ref{algorithm_scheme} with the double well 
potential~$f$ according to~\eqref{f_phi}. Assume that $\rho^{n} >0$ for all $n$. If the 
stabilization parameter $\beta$ in  \eqref{chem_pot_modelII_simp_nondim_timedis} is selected in accordance with 
\begin{equation}
\beta \geq 1 \label{beta}
\end{equation}
then the scheme in Algorithm~\ref{algorithm_scheme} is:
\begin{itemize}
\item[(i)] \textbf{unconditionally energy stable}: for all $n=1,\ldots, N-1$ the following discrete energy-dissipation relation holds:
\begin{equation}
\label{energy_dissip_disc}
\begin{aligned}
E^{n+1} - E^{n} \leq & -\frac{\Delta t}{\mathrm{Re}}  \|\nabla  \mathbf{v}^{n+1} \|_{L^{2}}^{2} - \frac{\Delta t(1+\lambda)}{\mathrm{Re}} \| \nabla \cdot \mathbf{v}^{n+1} \|_{L^{2}}^2 
\\
&- \frac{\Delta t}{\mathrm{Pe}\, \mathrm{We}} \| \nabla \mu^{n+1} \|_{L^{2}}^2 
 - \frac{1}{2} \rho^{n} \| \mathbf{v}^{n+1} - \mathbf{v}^{n} \|_{L^{2}}^2  
\\ 
& - \frac{\mathrm{Cn}}{2 \,\mathrm{We}} \| \nabla (\phi^{n+1} - \phi^{n}) \|_{L^{2}}^2 
- \frac{ \beta-1}{\mathrm{We} \, \mathrm{Cn}}   \|\phi^{n+1} - \phi^{n}\|_{L^{2}}^2 \leq 0
\end{aligned}
\end{equation}
independent of the time-step size $\Delta t > 0$.
\item[(ii)] \textbf{mass and phase conserving}: for all $n=1,\ldots, N-1$ there holds
\begin{equation}
 \label{preserve_mass}
\begin{aligned}
\int_{\Omega}  \rho^{n+1}\, d \Omega  = &\int_{\Omega}  \rho^{0}\, d \Omega, \quad \int_{\Omega} \hat{\rho}^{n+1} \, d \Omega  = \int_{\Omega} \hat{\rho}^{0} \, d \Omega, \quad \text{and}  \\
&\int_{\Omega} \phi^{n+1} \, d \Omega  = \int_{\Omega}  \phi^{0} \, d \Omega.
\end{aligned}
\end{equation}
\end{itemize}
\end{theorem}
\begin{proof} 
~\newline
\textit{(i)} Our proof of the discrete energy dissipation relation \eqref{energy_dissip_disc} closely follows the derivation for the time-continuous case in the proof of Theorem.~\ref{thm_1}. Using the definition of the energy in \eqref{energy_nondim}, the discrete energy can be written as
\begin{equation}
E^n = \int_{\Omega} \left( \frac{1}{2} \rho^n |\mathbf{v}^n|^2 + \frac{1}{\mathrm{We} \, \mathrm{Cn}} f(\phi^n) + \frac{\mathrm{Cn}}{2 \, \mathrm{We}} |\nabla \phi^n|^2 + \frac{1}{\mathrm{Fr}^2}  \hat{\rho}^n y\right) \, d\Omega. \label{energy_nondim_disc}
\end{equation}
For the difference in discrete energies at $t^{n+1}$ and $t^n$ it then follows that
\begin{equation}
\label{disc_energy_1}
\begin{aligned}
\hspace{-0.2cm} E^{n+1}-E^{n} &= \int_{\Omega}  \frac{1}{2} \left( \rho^{n+1} |\mathbf{v}^{n+1}|^2  - \rho^{n} |\mathbf{v}^{n}|^2\right)\, d \Omega  \\ 
&+ \int_{\Omega} \frac{1}{\mathrm{We} \, \mathrm{Cn}} \left( f(\phi^{n+1}) - f(\phi^{n})\right) \, d \Omega  \\
&+\int_{\Omega}  \frac{\mathrm{Cn}}{2\, \mathrm{We}} \left( |\nabla \phi^{n+1}|^2 - |\nabla \phi^{n}|^2 \right) \, d \Omega \\
& + \int_{\Omega} \frac{1}{\mathrm{Fr}^2} \left( \hat{\rho}^{n+1} - \hat{\rho}^{n} \right)y \, d \Omega. 
\end{aligned}
\end{equation}
For the first term in \eqref{disc_energy_1}, it holds that
\begin{align}
\int_{\Omega}  \frac{1}{2} \left( \rho^{n+1} |\mathbf{v}^{n+1}|^2 - \rho^{n} |\mathbf{v}^{n}|^2\right)\, d \Omega &= \nonumber \\
&\hspace{-3.5cm} \int_{\Omega}  \frac{1}{2} (\rho^{n+1} - \rho^{n}) |\mathbf{v}^{n+1}|^2\, d \Omega + \int_{\Omega} \frac{1}{2} \rho^{n}  \left( |\mathbf{v}^{n+1}|^2 - |\mathbf{v}^{n}|^2\right) \, d \Omega \label{first_term}
\end{align}
by adding and subtracting $\rho^{n} |\mathbf{v}^{n+1}|^{2}/2$. Using \eqref{density_mass_modelII_simp_nondim_timedis} and invoking integration by parts on $ (\rho^{n+1} - \rho^{n}) |\mathbf{v}^{n+1}|^2$, we obtain
\begin{equation}
\int_{\Omega}  \frac{1}{2} (\rho^{n+1} - \rho^{n}) |\mathbf{v}^{n+1}|^2\, d \Omega = \int_{\Omega} \Delta t \rho^{n} \left(\mathbf{v}^{n} \cdot \nabla \right) \mathbf{v}^{n+1} \cdot \mathbf{v}^{n+1} \, d \Omega. \nonumber
\end{equation}
Moreover, applying the algebraic identity $\left( a^2 - b^2 \right) = 2 a (a-b) - (a-b)^2$
to $\rho^{n}  \left( |\mathbf{v}^{n+1}|^2 - |\mathbf{v}^{n}|^2\right)/2$ gives
\begin{align}
\int_{\Omega} \frac{1}{2} \rho^{n}  \left( |\mathbf{v}^{n+1}|^2 - |\mathbf{v}^{n}|^2\right) \, d \Omega 
&= \nonumber \\ 
&\hspace{-2.5cm}  \int_{\Omega}  \rho^{n} \mathbf{v}^{n+1} \cdot \left( \mathbf{v}^{n+1} - \mathbf{v}^{n} \right)\, d \Omega  - \int_{\Omega}  \frac{1}{2} \rho^{n} | \mathbf{v}^{n+1} - \mathbf{v}^{n} |^2 \, d \Omega. \nonumber
\end{align}
Hence, the identity \eqref{disc_energy_1} can be recast into
\begin{equation}
\begin{aligned}
E^{n+1}-E^{n} = &
\int_{\Omega}  \Big( \Delta t \rho^{n} (\mathbf{v}^{n} \cdot \nabla) \mathbf{v}^{n+1} \cdot \mathbf{v}^{n+1} + \rho^{n} \mathbf{v}^{n+1} \cdot \left( \mathbf{v}^{n+1}  - \mathbf{v}^{n} \right) \Big) \, d \Omega 
\\
&  - \frac{1}{2} \rho^{n} \| \mathbf{v}^{n+1} - \mathbf{v}^{n} \|_{L^{2}}^2
+ \int_{\Omega} \frac{1}{\mathrm{We} \, \mathrm{Cn}} \left( f(\phi^{n+1}) - f(\phi^{n})\right) \, d \Omega 
\\
& + \frac{\mathrm{Cn}}{2\, \mathrm{We}} \left( \| \nabla \phi^{n+1} \|_{L^{2}}^2 -  \| \nabla \phi^{n} \|_{L^{2}}^2\right) 
+ \int_{\Omega} \frac{1}{\mathrm{Fr}^2} \frac{d \hat{\rho}^{n}}{d \phi^{n}} (\phi^{n+1} -\phi^{n}) y \, d \Omega \label{disc_energy_2}
\end{aligned}
\end{equation}
Note that the ultimate terms in \eqref{disc_energy_2} and in \eqref{disc_energy_1} coincide by virtue of the identities:
\begin{equation}
\hat{\rho}^{n+1} - \hat{\rho}^{n} = \frac{1}{2} \left( 1+\frac{\rho_{2}}{\rho_{1}}\right) (\phi^{n+1} -\phi^{n}) = \frac{d \hat{\rho}^{n}}{d \phi^{n}} (\phi^{n+1} -\phi^{n}).  \label{hat_rho_prop}
\end{equation} 
Next, we regard the time-stepping scheme \eqref{phase_modelII_simp_nondim_timedis}--\eqref{div_v_modelII_simp_nondim_timedis}. By multiplying the phase equation \eqref{phase_modelII_simp_nondim_timedis} with 
\begin{equation}
\frac{\Delta t}{\mathrm{We}}\bigg(\mu^{n+1} + \frac{1}{\hat{\rho}^{n}} \frac{d \hat{\rho}^{n}}{ d \phi^{n}} p^{n+1}\bigg) + \frac{\Delta t}{\mathrm{Fr}^2} \frac{d \hat{\rho}^{n}}{d \phi^{n}}  y
\end{equation}
integrating over the domain and invoking integration by parts, we obtain
\begin{equation}
\label{phase_test}
\begin{aligned} 
&\int_{\Omega}  \frac{1}{\mathrm{We}} \left(\mu^{n+1} + \frac{1}{\hat{\rho}^{n}} \frac{d \hat{\rho}^{n}}{ d \phi^{n}} p^{n+1} \right)  \Big( (\phi^{n+1}-\phi^{n}) + \Delta t\nabla \cdot  (\phi^{n} \mathbf{v} ^{n+1}) \Big)  \, d \Omega  \\
&+ \int_{\Omega} \frac{1}{\mathrm{Fr}^2} \frac{d \hat{\rho}^{n}}{d \phi^{n}}  y \Big( (\phi^{n+1}-\phi^{n}) + \Delta t\nabla \cdot  (\phi^{n} \mathbf{v} ^{n+1}) \Big)  \, d \Omega  \\
& = - \frac{\Delta t}{\mathrm{Pe}\, \mathrm{We}} \| \nabla \mu^{n+1} \|_{L^{2}}^2 
+ \int_{\Omega}\left( \frac{1}{\mathrm{We}} \frac{1}{\hat{\rho}^{n}} \frac{d \hat{\rho}^{n}}{ d \phi^{n}} p^{n+1} +\frac{1}{\mathrm{Fr}^2} \frac{d \hat{\rho}^{n}}{d \phi^{n}}  y \right) \frac{\Delta t}{\mathrm{Pe}} \triangle \mu^{n+1}  \, d \Omega.
\end{aligned}
\end{equation}
Multiplying \eqref{chem_pot_modelII_simp_nondim_timedis} by $\displaystyle -(\phi^{n+1} - \phi^{n})/\mathrm{We}$, integrating over the domain and using integration by parts, we deduce:
\begin{equation}
\label{mu_test}
\begin{aligned}
 &- \int_{\Omega} \frac{1}{\mathrm{We}} \mu^{n+1} (\phi^{n+1} - \phi^{n}) \, d \Omega  \\
&  \quad= - \frac{\beta}{\mathrm{We} \, \mathrm{Cn}}  \|\phi^{n+1} - \phi^{n}\|_{L^{2}}^2 - \int_{\Omega} \frac{1}{\mathrm{We} \, \mathrm{Cn}} f'(\phi^{n}) (\phi^{n+1} - \phi^{n})  \, d \Omega   
- \frac{\mathrm{Cn}}{\mathrm{We}} \| \nabla \phi^{n+1} \|_{L^{2}}^2 
\\
&\phantom{=}\,\quad\,
+  \int_{\Omega} \frac{\mathrm{Cn}}{\mathrm{We}} (\nabla \phi^{n+1} \cdot \nabla \phi^{n} )\, d \Omega  
+ \int_{\Omega} \frac{1}{\mathrm{We} }  (\phi^{n+1} - \phi^{n})  \frac{1}{\hat{\rho}^{n}} \frac{d \hat{\rho}^{n}}{d \phi^{n}}  p^{n+1} \, d \Omega  
\end{aligned}
\end{equation}
Similary, multiplication of~\eqref{momentum_modelII_simp_nondim_timedis} by $\Delta t\, {\mathbf{v}}^{n+1}$, integrating over the domain and invoking integration by parts yields:
\begin{equation}
\label{velocity_test}
\begin{aligned}
&\int_{\Omega }\Big( \rho^{n} \mathbf{v}^{n+1} \cdot \left( \mathbf{v}^{n+1}  - \mathbf{v}^{n} \right)  +\Delta t \rho^{n} (\mathbf{v}^{n} \cdot \nabla) \mathbf{v}^{n+1} \cdot \mathbf{v}^{n+1} \Big) \, d \Omega \\
&\quad= -\int_{\Omega } \frac{\Delta t}{\mathrm{We}} \left( \nabla p^{n+1} + \phi^{n} \nabla \mu^{n+1}  - \frac{1}{\hat{\rho}^{n}} \frac{d \hat{\rho}^{n}}{d\phi^{n}} p^{n+1} \nabla \phi^{n} \right) \cdot {\mathbf{v}}^{n+1} \, d \Omega \\
&\quad\phantom{=}\,\,-\frac{\Delta t}{\mathrm{Re}}  \|\nabla  \mathbf{v}^{n+1} \|_{L^{2}}^{2} - \frac{\Delta t(1+\lambda)}{\mathrm{Re}} \| \nabla \cdot \mathbf{v}^{n+1} \|_{L^{2}}^2 - \int_{\Omega } \frac{\Delta t}{\mathrm{Fr}^{2}} \hat{\rho}^{n} {\boldsymbol{\jmath}} \cdot {\mathbf{v}}^{n+1} \, d \Omega 
\end{aligned}
\end{equation}
Finally, upon multiplying~\eqref{div_v_modelII_simp_nondim_timedis} by
\begin{equation}
\displaystyle - \frac{\Delta t}{\alpha} \left( \frac{1}{\mathrm{We}} \frac{1}{\hat{\rho}^{n}} \frac{d \hat{\rho}^{n}}{d \phi^{n}}  p^{n+1} + \frac{1}{\mathrm{Fr}^2} \frac{d \hat{\rho}^{n}}{d \phi^{n}}  y \right)
\end{equation}
and integrating over the domain, we obtain: 
\begin{align}
& - \int_{\Omega } \frac{\Delta t}{\alpha} \left( \frac{1}{\mathrm{We}} \frac{1}{\hat{\rho}^{n}} \frac{d \hat{\rho}^{n}}{d \phi^{n}}  p^{n+1} + \frac{1}{\mathrm{Fr}^2} \frac{d \hat{\rho}^{n}}{d \phi^{n}}  y \right) (\nabla \cdot \mathbf{v}^{n+1})  \, d \Omega \nonumber \\
& \quad=  - \int_{\Omega }  \frac{\Delta t}{\mathrm{Pe}} \left( \frac{1}{\mathrm{We}} \frac{1}{\hat{\rho}^{n}} \frac{d \hat{\rho}^{n}}{d \phi^{n}}  p^{n+1} + \frac{1}{\mathrm{Fr}^2} \frac{d \hat{\rho}^{n}}{d \phi^{n}}  y \right) \triangle \mu^{n+1} \, d \Omega \label{div_test}
\end{align}
By collecting the results in \eqref{phase_test}--\eqref{div_test}, we obtain the identity:
\begin{equation}
\label{disc_energy_3}
\begin{aligned}
 & \int_{\Omega} \Big( \rho^{n} \mathbf{v}^{n+1} \cdot \left( \mathbf{v}^{n+1}  - \mathbf{v}^{n} \right) + \Delta t \rho^{n} (\mathbf{v}^{n} \cdot \nabla) \mathbf{v}^{n+1} \cdot \mathbf{v}^{n+1} \Big)\, d \Omega  \\
 &+ \int_{\Omega}  \frac{1}{\mathrm{Fr}^2} \frac{d \hat{\rho}^{n}}{d \phi^{n}} (\phi^{n+1} -\phi^{n}) y \, d \Omega  \\
&\quad=-\frac{\Delta t}{\mathrm{Re}}  \|\nabla  \mathbf{v}^{n+1} \|_{L^{2}}^{2} - \frac{\Delta t(1+\lambda)}{\mathrm{Re}} \| \nabla \cdot \mathbf{v}^{n+1} \|_{L^{2}}^2 - \frac{\Delta t}{\mathrm{Pe}\, \mathrm{We}} \| \nabla \mu^{n+1} \|_{L^{2}}^2  \\
 &\quad\phantom{=\,\,}- \frac{\mathrm{Cn}}{\mathrm{We}} \| \nabla \phi^{n+1} \|_{L^{2}}^2 +  \int_{\Omega} \frac{\mathrm{Cn}}{\mathrm{We}} (\nabla \phi^{n+1} \cdot \nabla \phi^{n} )\, d \Omega \\
 &\quad\phantom{=\,\,} - \int_{\Omega} \frac{1}{\mathrm{We} \, \mathrm{Cn}} f'(\phi^{n}) (\phi^{n+1} - \phi^{n})   \, d \Omega - \frac{\beta}{\mathrm{We} \, \mathrm{Cn}}  \|\phi^{n+1} - \phi^{n}\|_{L^{2}}^2    \\
 &\quad\phantom{=\,\,} + \int_{\Omega}  \frac{\Delta t}{\mathrm{We}} \left(  - \nabla p^{n+1} \cdot \mathbf{v}^{n+1} + \frac{1}{\hat{\rho}^{n}} \frac{d \hat{\rho}^{n}}{d\phi^{n}} p^{n+1}  ( - \phi^{n} \nabla \cdot \mathbf{v}^{n+1} + \frac{1}{\alpha} \nabla \cdot \mathbf{v}^{n+1} )\right) \, d \Omega  \\
&\quad\phantom{=\,\,} + \int_{\Omega} \frac{\Delta t}{\mathrm{Fr}^{2}} \left( -\hat{\rho}^{n} {\boldsymbol{\jmath}} \cdot \mathbf{v}^{n+1} - \frac{d \hat{\rho}^{n}}{d \phi^{n}} y \left(\nabla \cdot (\phi^{n} \mathbf{v}^{n+1}) - \frac{1}{\alpha} \nabla \cdot \mathbf{v}^{n+1}  \right)  \right) \, d \Omega.  
\end{aligned}
\end{equation}
The penultimate and ultimate terms in the right member of~\eqref{disc_energy_3} vanish
by virtue of relation~\eqref{relation2}, in a similar manner as their continuous counterparts in \eqref{vanish_integrals}.
By replacing the first and the last terms in \eqref{disc_energy_2} in accordance with \eqref{disc_energy_3}, we obtain
\begin{align}
E^{n+1} - E^{n} = & -\frac{\Delta t}{\mathrm{Re}}  \|\nabla  \mathbf{v}^{n+1} \|_{L^{2}}^{2} - \frac{\Delta t(1+\lambda)}{\mathrm{Re}} \| \nabla \cdot \mathbf{v}^{n+1} \|_{L^{2}}^2 - \frac{\Delta t}{\mathrm{Pe}\, \mathrm{We}} \| \nabla \mu^{n+1} \|_{L^{2}}^2 \nonumber \\
&  - \frac{1}{2} \rho^{n} \| \mathbf{v}^{n+1} - \mathbf{v}^{n} \|_{L^{2}}^2  - \frac{\mathrm{Cn}}{2 \,\mathrm{We}} \| \nabla (\phi^{n+1} - \phi^{n}) \|_{L^{2}}^2 \nonumber \\
& + \int_{\Omega} \frac{1}{\mathrm{We} \, \mathrm{Cn}} \Big( -f'(\phi^{n}) (\phi^{n+1} - \phi^{n}) +  f(\phi^{n+1}) - f(\phi^{n}) \Big) \, d \Omega\nonumber \\
&  - \frac{\beta}{\mathrm{We} \, \mathrm{Cn}}  \|\phi^{n+1} - \phi^{n}\|_{L^{2}}^2. \label{disc_energy_4} 
\end{align}
To derive \eqref{disc_energy_4}, we have also combined the sum of the fourth term in the right hand side 
of~\eqref{disc_energy_2} with the fourth and the fifth terms in the right hand side of~\eqref{disc_energy_3}
according to:
\begin{align}
\frac{\mathrm{Cn}}{2\, \mathrm{We}} &\left( \| \nabla \phi^{n+1} \|_{L^{2}}^2 -  \| \nabla \phi^{n} \|_{L^{2}}^2\right) - \frac{\mathrm{Cn}}{\mathrm{We}} \| \nabla \phi^{n+1} \|_{L^{2}}^2 \nonumber \\
 + \int_{\Omega} &\frac{\mathrm{Cn}}{\mathrm{We}} (\nabla \phi^{n+1} \cdot \nabla \phi^{n} )\, d \Omega = - \frac{\mathrm{Cn}}{2 \,\mathrm{We}} \| \nabla (\phi^{n+1} - \phi^{n}) \|_{L^{2}}^2. 
\end{align}
Finally, we use Taylor expansion on the double-well function $f(\phi)$ to obtain the identity:
\begin{equation}
 f(\phi^{n+1}) - f(\phi^{n}) = f'(\phi^{n}) (\phi^{n+1} - \phi^{n}) + \frac{f''(\xi^{n})}{2}  (\phi^{n+1} - \phi^{n})^2 \nonumber
\end{equation}
for some $\xi^{n} \in [\phi^{n}, \phi^{n+1}]$ and for all $\phi^{n+1}$. For the double-well function in \eqref{f_phi}, it holds that: 
\begin{equation}
    \max\limits_{\phi \in \mathbb{R}} |f''(\phi)| \leq 2.
\end{equation} 
From \eqref{disc_energy_4}, we then infer the following bound:
\begin{equation}
\label{disc_energy_5} 
\begin{aligned}
E^{n+1} - E^{n} \leq&  -\frac{\Delta t}{\mathrm{Re}}  \|\nabla  \mathbf{v}^{n+1} \|_{L^{2}}^{2} - \frac{\Delta t(1+\lambda)}{\mathrm{Re}} \| \nabla \cdot \mathbf{v}^{n+1} \|_{L^{2}}^2 - \frac{\Delta t}{\mathrm{Pe}\, \mathrm{We}} \| \nabla \mu^{n+1} \|_{L^{2}}^2 \\
&  - \frac{1}{2} \rho^{n} \| \mathbf{v}^{n+1} - \mathbf{v}^{n} \|_{L^{2}}^2  - \frac{\mathrm{Cn}}{2 \,\mathrm{We}} \| \nabla (\phi^{n+1} - \phi^{n}) \|_{L^{2}}^2
\\
&- \frac{ \beta-1}{\mathrm{We} \, \mathrm{Cn}}   \|\phi^{n+1} - \phi^{n}\|_{L^{2}}^2 
 \leq 0 
\end{aligned}
\end{equation}
The bound \eqref{disc_energy_5} implies the desired discrete dissipation law \eqref{energy_dissip_disc}. Let us note that in the above proof we have not imposed any conditions on the time step $\Delta t>0$.

\noindent
\textit{(ii)} Integrating \eqref{density_mass_modelII_simp_nondim_timedis} over the domain $\Omega$, applying the divergence theorem and the homogeneous boundary condition for velocity, we obtain:
\begin{equation}
 \int_{\Omega}  ( \rho^{n+1} - \rho^{n})  \, d \Omega  = \int_{\Omega} -  \Delta t \nabla \cdot (\rho^{n} \mathbf{v}^{n}) \, d \Omega 
=  \int_{\partial \Omega}  - \Delta t \rho^{n} \mathbf{v}^{n}  \cdot \mathbf{n} \, dS = 0.  \label{rho_preserve}
\end{equation}
Similarly, by integrating \eqref{hat_rho_prop} over the domain, using \eqref{phase_modelII_simp_nondim_timedis} and \eqref{div_v_modelII_simp_nondim_timedis} together with the discrete version of \eqref{relation2} and applying the boundary conditions in \eqref{boun_con_disc}, we obtain the following sequence of identities:
\begin{equation} 
\label{hat_rho_preserve}
\begin{aligned}
 \int_{\Omega}  (\hat{\rho}^{n+1} - \hat{\rho}^{n}) \, d \Omega &= \frac{d \hat{\rho}^{n}}{d \phi^{n}}  \int_{\Omega}  (\phi^{n+1} -\phi^{n})  \, d \Omega  \\
 &= \frac{d \hat{\rho}^{n}}{d \phi^{n}}   \int_{\Omega}\Delta t\bigg(  -  \nabla \cdot  (\phi^{n} \mathbf{v} ^{n+1}) +\frac{1}{\mathrm{Pe}} \triangle \mu^{n+1}  \bigg)  \, d \Omega  \\
 &= \frac{d \hat{\rho}^{n}}{d \phi^{n}}   \int_{\Omega} \Delta t \bigg( - \nabla \cdot  (\phi^{n} \mathbf{v} ^{n+1}) +\frac{1}{\alpha} \nabla \cdot  \mathbf{v}^{n+1}  \bigg)  \, d \Omega  \\
 &  =  \int_{\Omega}  - \Delta t \nabla \cdot (\hat{\rho}^{n} \mathbf{v}^{n+1}) \, d \Omega 
=  \int_{\partial \Omega}  - \Delta t \hat{\rho}^{n} \mathbf{v}^{n+1}  \cdot \mathbf{n} \, dS = 0.
\end{aligned}
\end{equation}
The assertions in~\eqref{preserve_mass} follow by induction on \eqref{rho_preserve} and \eqref{hat_rho_preserve}.
\end{proof}

\begin{remark}
Compared to the continuous dissipation relation \eqref{energy_dissipation}, the discrete energy dissipation \eqref{disc_energy_5} has additional dissipation terms due to the underlying backward Euler method 
in Algorithm~\ref{algorithm_scheme} and the stabilization term, viz. $ - \frac{1}{2} \rho^{n} \| \mathbf{v}^{n+1} - \mathbf{v}^{n} \|_{L^{2}}^2 - \frac{\mathrm{Cn}}{2 \,\mathrm{We}} \| \nabla (\phi^{n+1} - \phi^{n}) \|_{L^{2}}^2$ and $ - \frac{\beta - 1}{\mathrm{We} \, \mathrm{Cn}} \|\phi^{n+1} - \phi^{n}\|_{L^{2}}^2$, respectively. Gravity does not contribute to the dissipation, neither in the continuous dissipation relation \eqref{energy_dissipation} nor in the time-discrete dissipation relation \eqref{disc_energy_5}. 
\end{remark}

\begin{remark}
A stabilization term similar to the one in~\eqref{chem_pot_modelII_simp_nondim_timedis} has been proposed by Shen, Yang and Wang\cite{SheYanWanCCP2013}. Their stabilization is motivated on the basis of the heuristic argument that it damps high-frequency or high wave-number modes in the numerical simulation, thus stabilizing the time-integration scheme and allowing for larger time steps. The proof of Theorem~\ref{thm2} conveys that the stabilization term in fact ensures that the energy-dissipation property of the quasi-incompressible NSCH system is retained in the time-discrete case.
\end{remark}

\begin{remark}
Theorem~\ref{thm2} is contingent on the premise that the mixture density~$\rho^n$ is positive for all~$n$. Otherwise,
the sign of the fourth term in~\eqref{energy_dissip_disc} reverses and the energy decay relation $E^{n+1}-E^n\leq{}0$
is not ensured. In addition, if positivity of $\rho^n$ is violated, then the energy~\eqref{energy_nondim_disc} does
not constitute a Lyapunov functional. For binary fluids with matching densities, $\rho_1=\rho_2$, positivity of $\rho^n$ is trivially satisfied. For non-matching densities, positivity of the mixture density is directly connected with the range of the phase variable~$\phi$, viz. compliance with $\phi\in[-1,1]$; cf. Equation~\eqref{rho_algebraic}. Accordingly,
assuming without loss of generality that $\rho_2<\rho_1$, it holds that $\rho\in[\rho_2,\rho_1]$. The conditions 
on the phase variable and the mixture density can be imposed a-priori by restricting $\phi(t)$ and $\rho(t)$ 
to the convex spaces
\begin{equation}
\label{eq:convexspaces}
\begin{aligned}
&\big\{\phi\in{}H^1(\Omega)\cap{}L^{\infty}(\Omega):|\phi|\leq{}1\text{ a.e. in }\Omega\big\}
\\
&\big\{\rho\in{}H^1(\Omega)\cap{}L^{\infty}(\Omega):\rho\in[\rho_2,\rho_1]\text{ a.e. in }\Omega\big\}
\end{aligned}
\end{equation}
Arrangement of the order parameter $\phi$ in a convex space instead of a general linear space has been investigated in the context of the Cahn--Hilliard equation with a non-smooth free energy.\cite{Blowey:1991qa} Numerical approximation methods for the Cahn--Hilliard equation in this setting have also been studied.\cite{Bosch:2014kx,Hintermuller:2011pi} However, it is not known if solutions to the NSCH system~\eqref{NSCH_nondim} subject to~\eqref{boun_con} 
and subject to initial conditions in~\eqref{eq:convexspaces} remain in~\eqref{eq:convexspaces} as time progresses.

Just as it is not known for~\eqref{NSCH_nondim} if it admits solutions that remain in~\eqref{eq:convexspaces} 
as time progresses, it is not known for the semi-discretization in Algorithm~\ref{algorithm_scheme}
if it has this property and, if so, under which circumstances this property is retained under spatial discretization.
In practice we observe that positivity of $\rho^n$ can be violated for large density ratios, on coarse meshes and at 
large time steps. By virtue of~\eqref{density_mass_modelII_simp_nondim_timedis}, for non-matched densities 
positivity of the mixture density $\rho^{n+1}$ according to Algorithm~\ref{algorithm_scheme} is ensured 
if the following (local) time-step restriction holds:
\begin{equation}
\label{eq:Deltatlimit}
(\Delta t)^{n+1} < \min\bigg\{\bigg\lfloor\frac{\rho^n}{\nabla\cdot(\rho^n\mathbf{v}^n)}\bigg\rfloor\text{ in }\Omega\bigg\}
\end{equation}
where $\lfloor\cdot\rfloor=\frac{1}{2}(\cdot)+\frac{1}{2}|\cdot|$ represents the non-negative part of a function~$(\cdot)$.
Hence, the time-integration scheme in Algorithm~\ref{algorithm_scheme} is energy stable if the time-step is set 
adaptively in accordance with~\eqref{eq:Deltatlimit}.
\end{remark}

\begin{remark}
In \eqref{chem_pot_modelII_simp_nondim_timedis}, any choice of $\beta \geq 1$ yields a non-unique splitting of the double-well potential $f(\phi)$ in a similar manner as originally  discussed by Eyre\cite{EyrPROC1998}. That is, for the splitting, the truncated $f(\phi)$ can be composed of a convex (contractive) and a concave (expansive) part according to $f = f_{c} - f_{e}$, where both functions $f_{c}$ and $f_{e}$ are convex. Particularly, in our computations we choose $\beta = 2$ which reduces the stabilization into the linearly-stabilized splitting also proposed by Eyre according to:
\begin{equation}\label{f_phi_split}
f_{c} - f_{e} = \left\{ \begin{array}{ll}
 \left(\phi^2 + \frac{1}{4}\right) - \left( -2 \phi - \frac{3}{4} \right), & \phi <-1\\
 \left(\phi^2 + \frac{1}{4}\right) - \left( \frac{3}{2} \phi^2 - \frac{1}{4} \phi^4 \right), & \phi \in [-1, 1]\\
 \left(\phi^2 + \frac{1}{4}\right) - \left( 2 \phi - \frac{3}{4} \right), & \phi > 1,
  \end{array} \right.
\end{equation} 
which leads to a linearly implicit algorithm.
\end{remark}

\begin{remark}
A fully discrete scheme can be obtained by applying a finite element method to the weak form \eqref{phase_modelII_simp_nondim_weak}--\eqref{div_v_modelII_simp_nondim_weak} with \eqref{density_alg_modelII_simp_nondim_a} and \eqref{density_alg_modelII_simp_nondim_b}. 
For analyses of fully discrete schemes, the interested reader is referred to Feng\cite{FenSINUM2006} and Giesselman \&~Preyer\cite{GiePryM2AN2015}.
 \end{remark}
%
\section{Numerical Experiments}
\label{section4}
In this section, we present numerical experiments using our linear semi-implicit energy-dissipative scheme according to  Algorithm~\ref{algorithm_scheme}. We consider $2$D test cases with matching and variable densities and show the unconditional stability in discrete energy dissipation for large time step sizes. We also consider a falling droplet to test the scheme with gravity. We set $\lambda=-2/3$ in all experiments. The dimensionless parameters that are used for the numerical simulations are listed in Table~\ref{example_par}.
\setcounter{table}{0}    
\begin{table}[h]
\begin{center}
\begin{tabular}{c|lllll}
\hline
\text{Example} $1$ & Cn=0.01 & We=0.45 & Pe=4.5 & Re=100 & Fr$^{-2}$=0 \\
\hline
\text{Example} $2$ & Cn=0.0625 & We=0.45 & Pe=4.5 & Re=100 & Fr$^{-2}$=0 \\
\hline
\text{Example} $3$ & Cn=0.03 & We=2.5 & Pe=20000 & Re=1.7 & Fr$^{-2}$=0.1  \\
\hline
\end{tabular}
\end{center}
\caption{Parameters for the test cases}
\label{example_par}
\end{table}
\par
In the computations, for the spatial discretization, we use $\mathbb{P}_{1}-\mathbb{P}_{1}$ finite-elements for the phase variables $\phi$ and $\mu$ and $\mathbb{P}_{2}-\mathbb{P}_{1}$ Taylor-Hood elements for velocity $\mathbf{v}$ and pressure $p$ on uniform meshes with square elements. For time discretization, the stabilizing term is taken $\beta = 2$ and we apply homogeneous Neumann and homogeneous Dirichlet boundary conditions to $\phi$, $\mu$ and $\mathbf{v}$, respectively, in accordance with~\eqref{boun_con}. We also take zero initial velocity $\mathbf{v}_{0}(\mathbf{x}) =\mathbf{0}$. The initial conditions for the phase variable, $\phi$ are determined on the basis of the considered initial fluid volumes. 

%
\subsection{Example $1$: Coalescence for various density ratios} 
The goal of this test case is to investigate the time-integration scheme for varying density ratios, $\rho_{2}\!:\!\rho_{1}$. To this end, we consider the coalescence of two sufficiently close but non-touching droplets in a domain $\Omega = (0,1)^2$ with time step size $\Delta t = 0.05$ for a matched density case with $\rho_{2}\!:\!\rho_{1} = 1\!:\!1$ and a very large density-ratio scenario with two heavier droplets set in a lighter ambient medium with $\rho_{2} \!:\!\rho_{1} = 1\!:\! 1000 $. One may note that 
the latter ratio is very high when compared with other numerical results in the literature. For both cases, we ignore the effect of gravity. The other parameters are listed in Table~\ref{example_par}. We set the initial condition for the 
pahse variable according to:
\begin{equation}
\label{eq:TC1phi0}
\phi_{0}(\mathbf{x}) = 1 - \sum_{i=1}^{2} \tanh \left( \frac{\sqrt{(x - x_{i} )^2 +(y-y_{i})^2} - r_{i}}{\mathrm{Cn} \, \sqrt{2}}\right)  
\end{equation}
with $r_1 = 0.25$ and $r_2 = 0.1$ and $(x_{1}, y_{1}) = (0.4, 0.5)$ and $(x_{2}, y_{2}) = (0.78, 0.5)$.
The initial condition~\eqref{eq:TC1phi0} represents two circular droplets with radii $r_1$ and~$r_2$ 
centered at~$(x_1,y_1)$ and~$(x_2,y_2)$. We cover~$\Omega$ with a spatial mesh composed of $128^2$ uniform 
elements, which provides a support structure for the finite element spaces detailed above.

Figure~\ref{matched_variable_1} and the top row of Figure~\ref{nonmatched_variable_1} present the evolution of the phase variable for the matched and the non-matched density cases, respectively, from $t=0.05$ to $50$. One can observe that 
for both the matched and non-matched density cases, the droplets coalesce and form a circular droplet. The initial coalescence can be attributed to diffusion, while the evolution to the stationary circular shape is due to capillary forces which have the effect of minimizing surface area.
%
\begin{figure}
\centering
\includegraphics[scale=1]{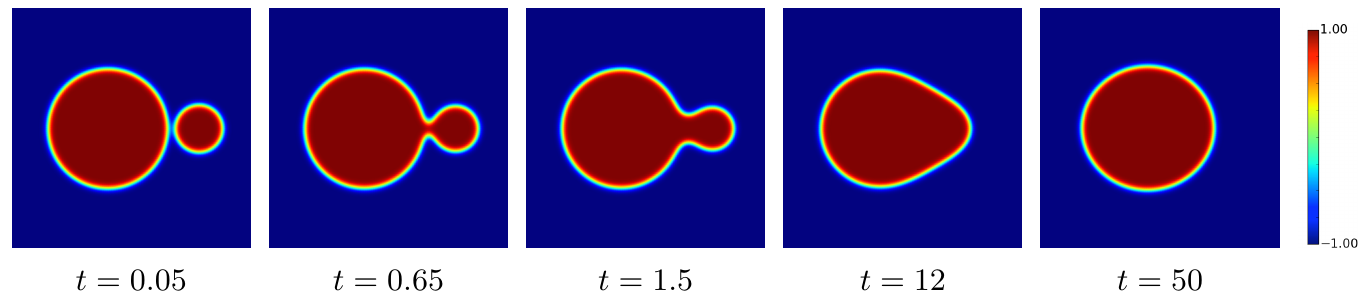}
\caption{\small{Droplet coalescence for matched densities: Evolution of phase variable, $\phi$ with density ratio $\rho_{2} \!:\! \rho_{1} = 1\!:\!1$,  from $t=0.05$ to $50$, Cn$=0.01$, We=$0.45$, Pe=$4.5$, Re=$100$, Fr$^{-2}=0$ and $\Delta t = 0.05$ }}
\label{matched_variable_1}
\end{figure}
\begin{figure}
\centering
\includegraphics[scale=1]{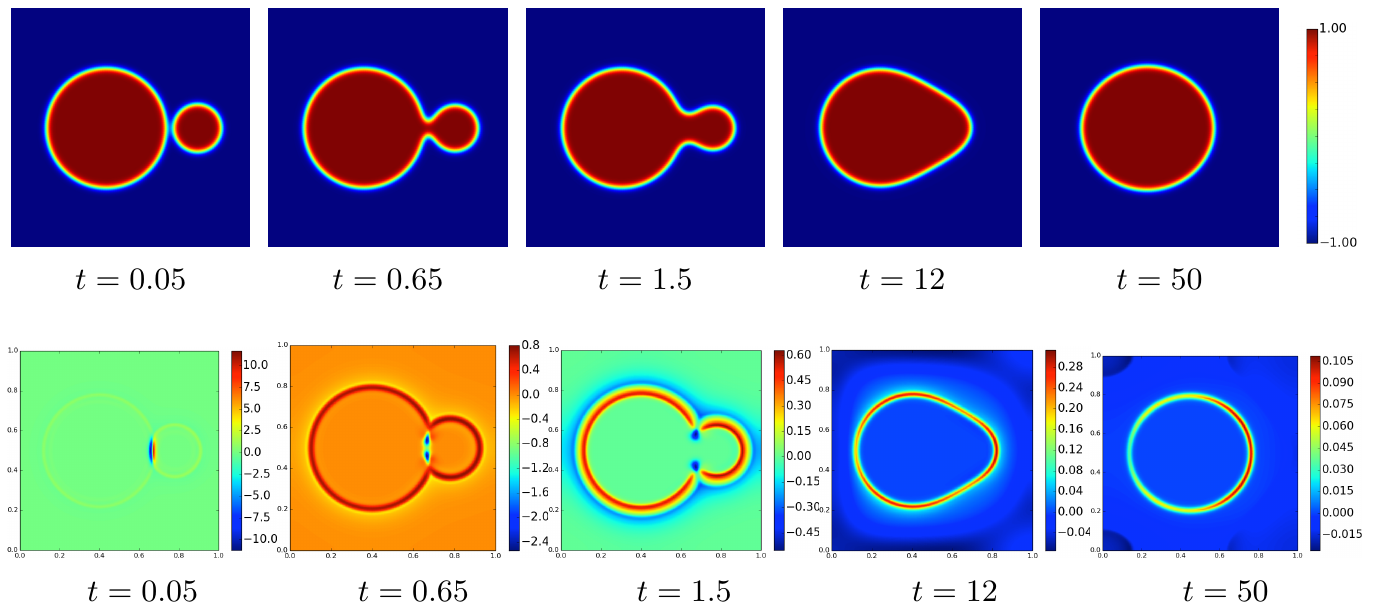}
\caption{\small{Droplet coalescence for non-matched densities: Evolution of phase variable, $\phi$ (top) and divergence of mixture velocity (bottom) with density ratio $\rho_{2} \!:\! \rho_{1} = 1\!:\!1000$, from $t=0.05$ to $50$, Cn$=0.01$, We=$0.45$, Pe=$4.5$, Re=$100$, Fr$^{-2}=0$ and $\Delta t = 0.05$}}
\label{nonmatched_variable_1}
\end{figure}
\par
The divergence of the mixture velocity for the non-matched density case is also displayed in Figure~\ref{nonmatched_variable_1}. The corresponding plots illustrate the quasi-incompressible behavior.
It can be observed that the divergence of the mixture velocity is non-zero, i.e. $\nabla \cdot \mathbf{v} \neq 0$, 
in the vicinity of the moving interface, where compressible mixing of the phases occurs. In the pure phases away from the interface, for which the individual components are incompressible, we indeed observe that the mixture velocity is solenoidal,
i.e. $\nabla \cdot \mathbf{v} = 0$. As opposed to the non-matching density case, for the matching-density case the incompressibility condition is satisfied up to discretization errors throughout the domain, in the pure phases as 
well as at the interface (results not displayed).

Figure~\ref{matched_variable_2} displays the evolution of the mass and energy corresponding to the
phase-field evolution in Figures~\ref{matched_variable_1}\nobreakdash--\ref{nonmatched_variable_1}. 
The results in the left panel confirm that the total mass is preserved in accordance with Theorem~\ref{thm2}. 
From the right panel one can observe that the energy is indeed non-increasing in accordance with the 
discrete energy-dissipation property of the time-integration scheme according to Theorem~\ref{thm2}.
\begin{figure}
\centering
\captionsetup[subfigure]{labelformat=empty}
    \subfloat[(a) Mass]{\includegraphics[width=.52\textwidth]{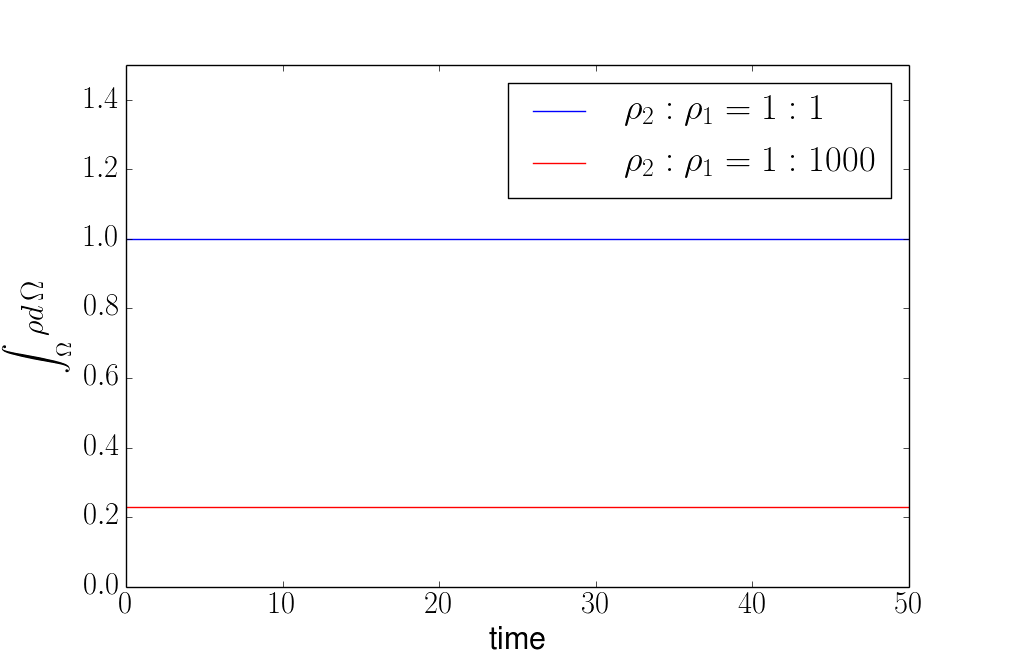}}
     \subfloat[(b) Energy]{\includegraphics[width=.52\textwidth]{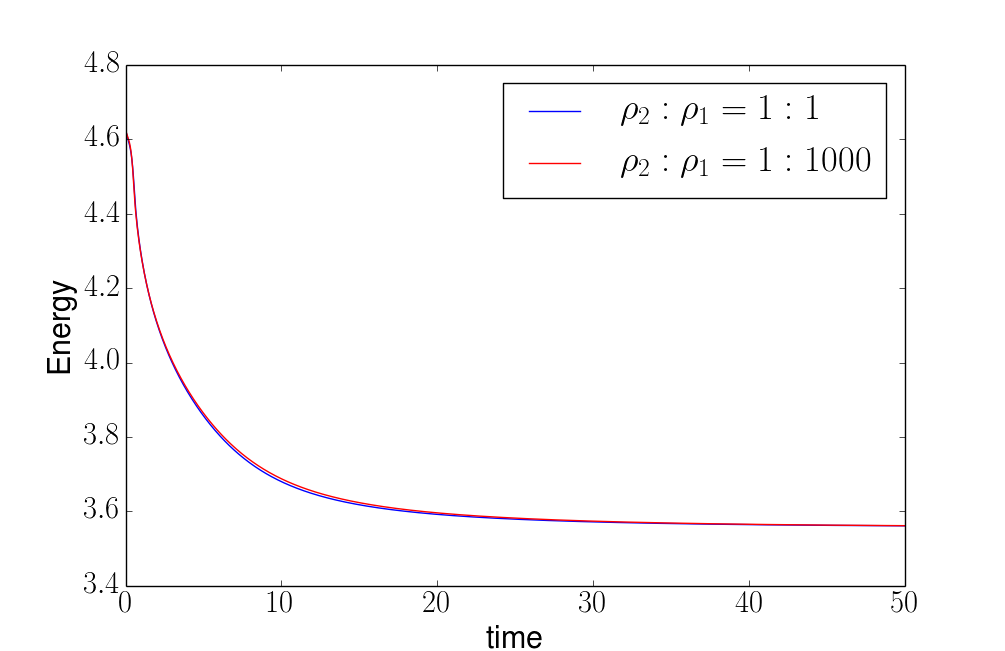}} \\
\caption{\small
Evolution of mass and energy for the droplet-coalescence test case for matched densities $\rho_{2} \!:\! \rho_{1} = 1\!:\!1$ 
and non-matched densities $\rho_{2} \!:\! \rho_{1} = 1\!:\!1000$.
\label{matched_variable_2}}
\end{figure}
\par
%

%
\subsection{Example $2$: Breakup of two droplets with large time step size} 
This test case serves to study the stability of the scheme in Algorithm~\ref{algorithm_scheme} for large time steps. We regard a rectangular domain $\Omega = (-2,2) \times (-4,4)$. The domain is covered with a mesh composed of $75 \times 150$ elements, which again supports finite-element approximation spaces as before. We consider a large density ratio
$\rho_2 : \rho_1 = 1 : 1000$ and a large time step $\Delta t = 0.5$. The considered setup pertains to breakup of two droplets of identical size connected by a thin liquid bridge. We take the dimensionless parameters as presented in 
Table~\ref{example_par}.
\par
Figure~\ref{dumbbell_1} shows snapshots of the evolution of the phase field. One can observe that the liquid bridge 
that initially connects the two droplets fissures under the effect of surface tension and the separate droplets 
subsequently evolve to a circular shape in time, during which their surface area decreases.
\begin{figure}[!t]
\centering
\includegraphics[scale=1]{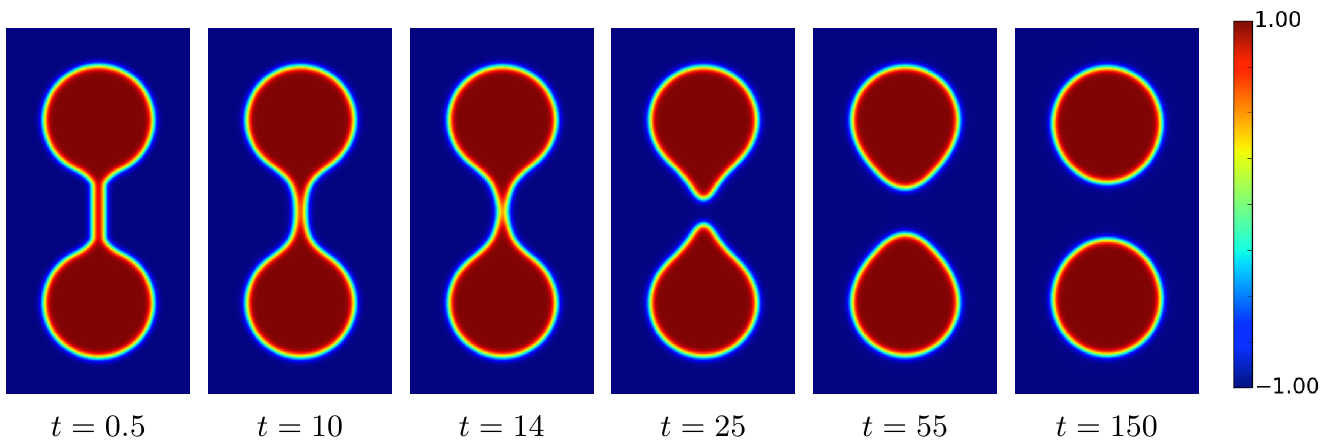}
\caption{\small{Breakup of a liquid bridge between two droplets for a large density ratio 
$\rho_{2} \!:\!\rho_{1} = 1\!:\! 1000$ and with large time-step size $\Delta t = 0.5$: 
Evolution of phase variable from $t=0.5$ to $150$ (Cn$=0.0625$, We=$0.45$, Pe=$4.5$, Re=$100$, Fr$^{-2}=0$).}}
\label{dumbbell_1}
\end{figure}
\begin{figure}[!b]
\centering
\captionsetup[subfigure]{labelformat=empty}
 \subfloat[(a)] {\includegraphics[width=.52\textwidth]{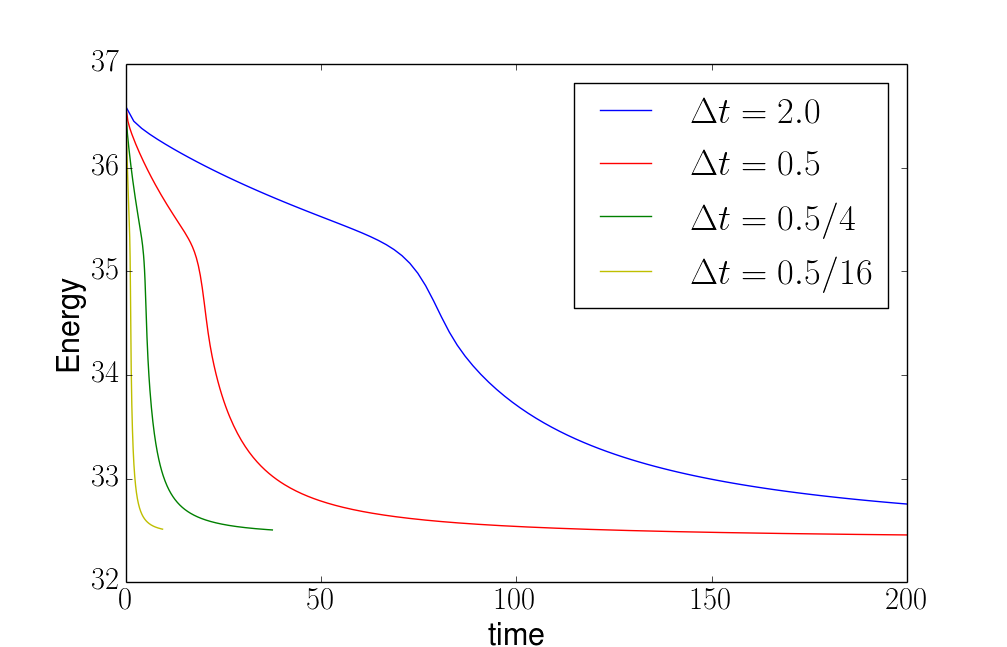}}
 \subfloat[(b) ] {\includegraphics[width=.52\textwidth]{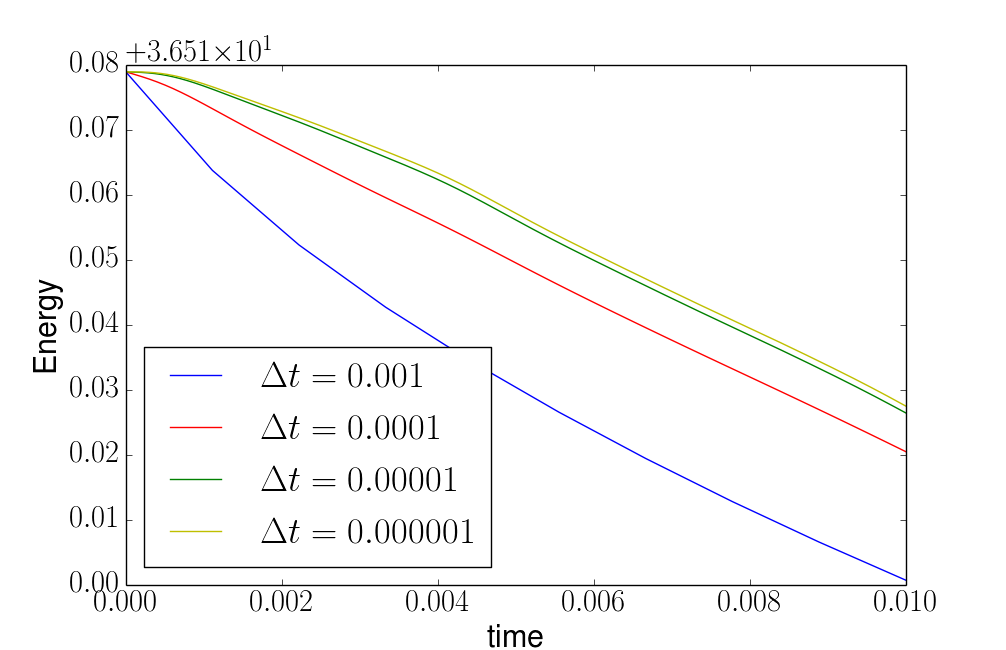}}\,
\caption{\small{
Energy evolution for breakup of a liquid bridge between two droplets for density ratio 
$\rho_{2} \!:\!\rho_{1} = 1\!:\! 1000$ and $\Delta t\in2^{\{-5,-3,-1,1\}}$  (a)
and  $\Delta t\in{}10^{\{-6,-5,-4,-3\}}$ (b).}}
\label{dumbbell_2}
\end{figure}
The evolution of the energy for various time-step sizes is presented in Figure~\ref{dumbbell_2}. In 
panel~(a), the energy evolution is presented for very large time step sizes. It is important to note that 
the time-integration scheme dissipates the energy even at such very large time steps, in accordance with 
Theorem~\ref{thm2}. The red dissipation curve in panel (a) corresponds to the simulation presented in 
Figure~\ref{dumbbell_1}. For all dissipation curves in panel (a), we observe a strong initial dissipation followed 
by a decrease of the dissipation rate. The initial dissipation rate signals the topological change by which the 
connected liquid volume separates into two disjoint liquid volumes. The asymptotic dissipation rate 
as $t \to \infty$ corresponds to the two droplets assuming their circular shape and the corresponding 
decrease in the surface area. Note that for diffuse interface problems of NSCH type
the considered time-step sizes in panel (a) are very large; typical time-step sizes for this type of problems reported 
in the literature are of the order of $10^{-3}$. By virtue of the fact that our time-integration scheme retains
its stability even at such large time-step sizes, it provides a solid basis for time adaptivity. In 
Figure~\ref{dumbbell_2}(b) we consider the energy dissipation of the time-integration scheme in the 
limit $\Delta t \to 0$. To this end, we plot the evolution of the energy versus time for 
$\Delta t = \{ 10^{-6}, \ldots, 10^{-3} \}$. One can observe that the energy evolution converges 
as $\Delta t \to 0$ and that small but finite dissipation remains in this limit, in accordance with 
the energy dissipation of the underlying quasi-incompressible NSCH equations. The total mass is preserved 
independent of the time-step size (results not displayed).
%
%
%
%
%
\subsection{Example $3$: Pinching droplet}
In this third test case, we explore the time-integration scheme subject to the effect of gravity. We consider a domain $\Omega = (-0.3, 0.3) \times (-0.8, 0.8)$ composed of $60 \times 200$ uniform square elements supporting the previously defined finite-element approximation spaces. Note that the computational domain is restricted to the right half of $\Omega$ in view of symmetry. The initial condition corresponds to a droplet attached to the top boundary:
\begin{align*}
\phi_{0}(\mathbf{x}) = - \tanh \left( \frac{\sqrt{x ^2 +(y- 2.9)^2} - 0.35}{\sqrt{2}\,\mathrm{Cn}}\right)  
\end{align*}
see Figure~\ref{falling_droplet_1}. The parameters are set in accordance with Table~\ref{example_par}. For this example, 
we set the density ratio $\rho_{2}\!:\!\rho_{1}=1\!:\!2.5$, because at large density ratios the condition $\rho^{n} >0$ 
can be violated for the considered large time-step size. 
\par
The evolution of the phase-field corresponding to the falling droplet is presented in Figure~\ref{falling_droplet_1}. Initially,  the droplet is attached to the upper wall. In time, due to gravity it moves downward and after detaching from the upper wall it continues its vertical movement as an almost circular shaped droplet. Upon reaching the bottom boundary, the droplet spreads under the Neumann boundary condition for the phase variable $\phi$, and assumes an ellipsoidal shape as $t \to \infty$, which corresponds to the equilibrium shape under gravity. It is to be noted that the homogeneous Neumann condition can be conceived of as a neutral wetting condition, i.e. the contact angle of the fluid--fluid interface with the solid wall corresponds to $\pi/2\,$.\cite{Jacqmin:2000kx,Yue:2011uq,Brummelen:2016qa}
\begin{figure}[!b]
\centering
\includegraphics[scale=1]{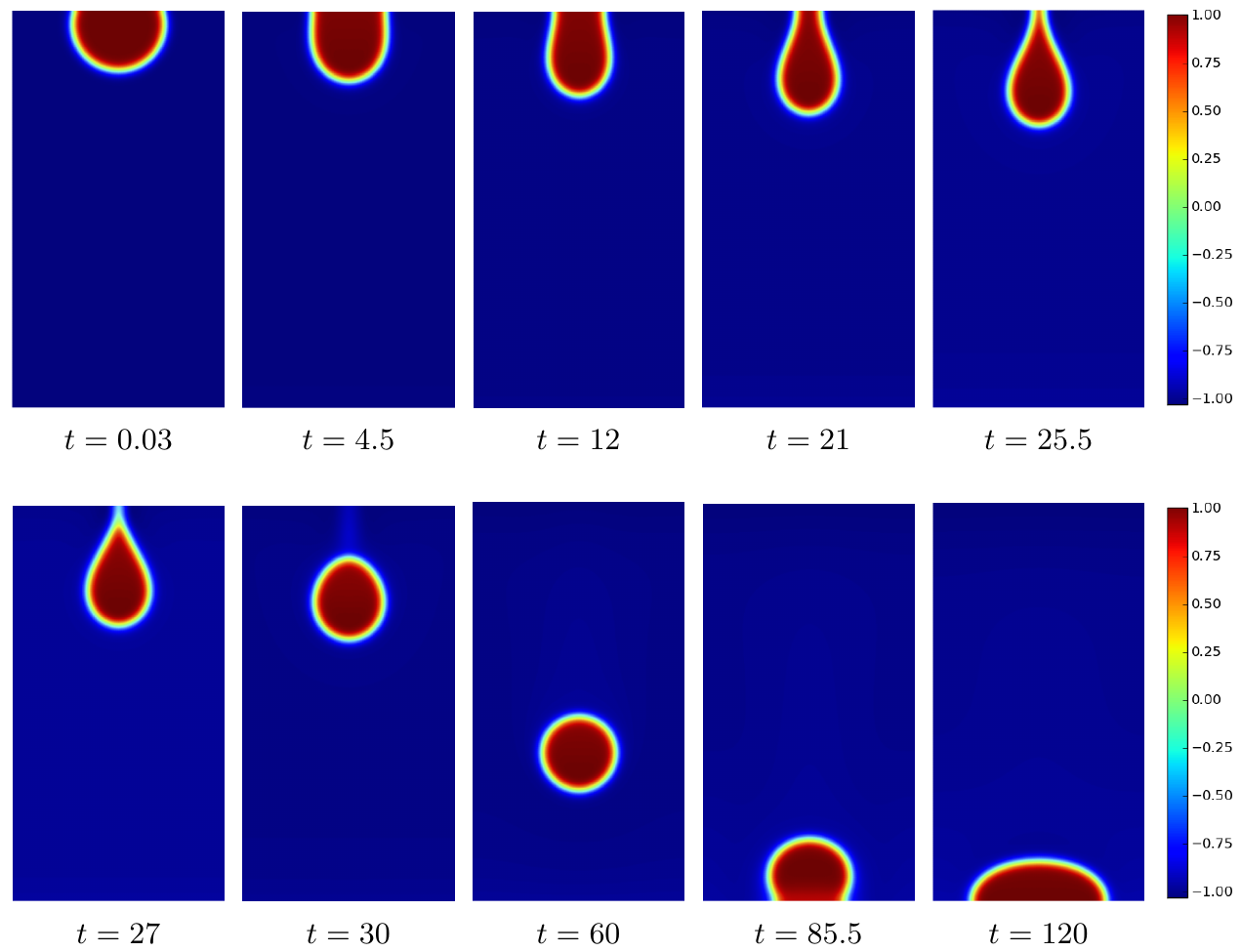}
\caption{\small{Evolution of the phase variable for a pinching droplet subject to gravity with 
density ratio $\rho_{2}\! :\!\rho_{1} = 1\!:\! 2.5$ from $t=0.03$ to $120$ (Cn$=0.03$, We=$2.5$, Pe=$2\times10^4$, 
Re=$1.7$, Fr$^{-2}=0.1$, $\Delta t = 0.03$).}}
\label{falling_droplet_1}
\end{figure}
%
%
\section{Conclusion}
\label{section5}
In this paper, we presented a new form of diffuse-interface Navier--Stokes-Cahn--Hilliard model 
for binary fluids with distinct densities. The model is based on the assumption that the fluids are individually incompressible but their partial mixture shows a quasi-incompressible character. Accordingly, the mass-averaged mixture velocity is generally non-solenoidal. We established an energy-dissipation relation for the quasi-incompressible NSCH system.
In addition, we presented a new linearly implicit time-integration scheme for the NSCH model 
and proved energy stability of this scheme independent of the time-step size. Moreover, we showed that the
time-integration scheme retains the mass and phase conservation properties of the underlying NSCH system. 
Central to the energy-dissipation property of the time-integration scheme is the use of two distinct but compatible definitions of the mixture density in the formulation.
\par
We conducted numerical experiments for breakup and coalescence of droplets and for a falling droplet subject to gravity. The numerical results confirm the energy-dissipation and mass-and-phase-conservation properties of the time-integration 
scheme for high density ratios and large time step sizes. In the computations pertaining to binary fluids with 
non-matching densities, we observed a non-zero velocity-divergence near the interface due to the mixing of the incompressible components of the fluid, in agreement with the quasi-incompressible behavior of the NSCH system.
\par
By virtue of the energy-stability of the presented time-integration scheme, the proposed scheme provides a solid basis for adaptive temporal refinement. We view the work presented in this
paper as preparatory to space-and-time adaptive approximation methods. In particular, adaptivity based on a-posteriori error estimates is essential to obtain efficient approximations for diffuse-interface type models with controllable accuracy.\cite{FenWuJCM2008,ZeeOdePruHawNMPDE2011,BarMulOrtSINUM2011,SimWuCMAME2015,WuZeeSimBru2017}
This is valuable for investigating the complicated multi-scale features that occur in many 
binary-fluid flows with capillarity, e.g. the formation of satellite droplets.

\section*{Acknowledgments}
The support of this work by NanoNextNL, a micro and nanotechnology consortium of the Government of the Netherlands and $130$ partners, is gratefully acknowledged. The authors would like to express their thanks to Gertjan van Zwieten for his help with the numerical implementations in the Nutils software (Nutils.org). The second author acknowledges support by the Netherlands Organisation for Scientific Research (NWO) via the Innovational Research Incentives Scheme (IRIS), Veni grant $639.031.033$. The authors are also grateful to the referees for their recommendations which have helped improve the paper.
\bibliography{BibFile,BibFileExtra}
\bibliographystyle{siam}
\clearpage
\tableofcontents
\end{document}

%% file: StandardSetOfCommands.tex

\newcommand{\tinyfbox}[1]{\raisebox{0.4ex}{\tiny{}\fbox{{#1}}}}

\newlength{\bigfboxsep}
\setlength{\bigfboxsep}{\fboxsep}
\addtolength{\bigfboxsep}{1ex}
\newcommand{\bigboxed}[1]{\setlength{\fboxsep}{\bigfboxsep}%
   \fbox{$\displaystyle{#1}$}}

\newcommand{\bigfbox}[1]{\setlength{\fboxsep}{\bigfboxsep}\fbox{#1}}

\newcommand{\sect}[1]{\S\ref{sec:#1}}
\newcommand{\eq}[1]{(\ref{eq:#1})}

\newcommand{\negquad}{\mspace{-18.0mu}}
\newcommand{\negqquad}{\negquad\negquad}
\newcommand{\negqqquad}{\negquad\negquad\negquad}
\newcommand{\negqqqquad}{\negquad\negquad\negquad\negquad}
\newcommand{\qqquad}{\qquad\quad}
\newcommand{\qqqquad}{\qquad\qquad}

\newcommand{\norm}[1]{{\|#1\|}}
\newcommand{\bignorm}[1]{\big\|#1\big\|}
\newcommand{\Bignorm}[1]{\Big\|#1\Big\|}
\newcommand{\biggnorm}[1]{\bigg\|#1\bigg\|}
\newcommand{\Biggnorm}[1]{\Bigg\|#1\Bigg\|}
\newcommand{\snorm}[1]{{|#1|}}
\newcommand{\bigsnorm}[1]{\big|#1\big|}
\newcommand{\Bigsnorm}[1]{\Big|#1\Big|}
\newcommand{\biggsnorm}[1]{\bigg|#1\bigg|}
\newcommand{\Biggsnorm}[1]{\Bigg|#1\Bigg|}
\newcommand{\enorm}[1]{{|\! |\! | #1 |\! |\! |}}
\newcommand{\bigenorm}[1]{{\big|\! \big|\! \big| #1 \big|\! \big|\! \big|}}
\newcommand{\Bigenorm}[1]{{\Big|\! \Big|\! \Big| #1 \Big|\! \Big|\! \Big|}}
\newcommand{\biggenorm}[1]{\bigg|\! \bigg|\! \bigg| #1 \bigg|\! \bigg|\! \bigg|}
\newcommand{\Biggenorm}[1]{\Bigg|\! \Bigg|\! \Bigg| #1\Bigg|\! \Bigg|\! \Bigg|}
\newcommand{\innerprod}[1]{(#1)}
\newcommand{\biginnerprod}[1]{\big(#1\big)}
\newcommand{\Biginnerprod}[1]{\Big(#1\Big)}
\newcommand{\bigginnerprod}[1]{\bigg(#1\bigg)}
\newcommand{\Bigginnerprod}[1]{\Bigg(#1\Bigg)}
\newcommand{\dual}[1]{\langle#1\rangle}
\newcommand{\bigdual}[1]{\big\langle#1\big\rangle}
\newcommand{\Bigdual}[1]{\Big\langle#1\Big\rangle}
\newcommand{\biggdual}[1]{\bigg\langle#1\bigg\rangle}
\newcommand{\Biggdual}[1]{\Bigg\langle#1\Bigg\rangle}
\newcommand{\jump}[1]{{[\! [#1 ]\! ]}}
\newcommand{\bigjump}[1]{{\big[\! \big[ #1 \big]\! \big]}}
\newcommand{\Bigjump}[1]{{\Big[\! \Big[ #1 \Big]\! \Big]}}
\newcommand{\biggjump}[1]{{\bigg[\! \bigg[\! #1\bigg]\! \bigg]}}
\newcommand{\Biggjump}[1]{{\Bigg[\! \Bigg[\! #1\Bigg]\! \Bigg]}}
\newcommand{\vsizeex}[2]{\resizebox{\width}{#1ex}{$#2$}}
\newcommand{\bracedot}{\vsizeex{1.2}{(}\mspace{-2mu}\raisebox{-0.15ex}{$\cdot$}\mspace{-2mu}\vsizeex{1.2}{)}}

\newcommand{\complexset}{\mathbb{C}}  
\newcommand{\realset}{\mathbb{R}} 
\newcommand{\polyset}{\mathbb{P}}
\newcommand{\nozero}{\setminus\{0\}}
\newcommand{\Dir}{{\mathscr{D}}}
\newcommand{\Neu}{{\mathscr{N}}}

\newcommand{\ds}[1]{{\displaystyle{#1}}}
\newcommand{\ts}[1]{{\textstyle{#1}}}

\newcommand{\half}{\frac{1}{2}}
\newcommand{\thalf}{\tfrac{1}{2}}

\newcommand{\divid}{\div}
\renewcommand{\div}{\operatorname{div}}
\providecommand{\grad}{\nabla}
\renewcommand{\grad}{\nabla}
\newcommand{\curl}{\operatorname{curl}}
\newcommand{\rot}{\operatorname{rot}}
\newcommand{\Ext}{\operatorname{Ext}}
\newcommand{\Id}{{I\mspace{-2mu}d}}
\newcommand{\interior}{\operatorname{int}}
\newcommand{\meas}{\operatorname{meas}}
\newcommand{\Res}{\operatorname{Res}}
\providecommand{\supp}{\operatorname{supp}}
\newcommand{\tr}{\operatorname{tr}}
\newcommand{\Span}{\operatorname{Span}}

\newcommand{\pd}{\partial}
\newcommand{\dd}{\mathrm{d}}
\newcommand{\DD}{\mathrm{D}}
\newcommand{\mo}{{\mspace{2.0mu}\text{{}-{}}1}}
\newcommand{\TT}{{\mspace{-2.0mu}\mathsf{T}\mspace{1.0mu}}}
\newcommand{\mT}{{\mspace{2.0mu}\text{{}-{}}\mathsf{T}}}

\newcommand{\mbb}[1]{\mathbb{#1}}

\newcommand{\bs}[1]{\boldsymbol{#1}} 
\newcommand{\bsa}{\bs{a}}
\newcommand{\bsb}{\bs{b}}
\newcommand{\bsc}{\bs{c}}
\newcommand{\bsd}{\bs{d}}
\newcommand{\bse}{\bs{e}}
\newcommand{\bsf}{\bs{f}}
\newcommand{\bsg}{\bs{g}}
\newcommand{\bsh}{\bs{h}}
\newcommand{\bsi}{\bs{i}}
\newcommand{\bsj}{\bs{j}}
\newcommand{\bsk}{\bs{k}}
\newcommand{\bsl}{\bs{l}}
\newcommand{\bsm}{\bs{m}}
\newcommand{\bsn}{\bs{n}}
\newcommand{\bso}{\bs{o}}
\newcommand{\bsp}{\bs{p}}
\newcommand{\bsq}{\bs{q}}
\newcommand{\bsr}{\bs{r}}
\newcommand{\bss}{\bs{s}}
\newcommand{\bst}{\bs{t}}
\newcommand{\bsu}{\bs{u}}
\newcommand{\bsv}{\bs{v}}
\newcommand{\bsw}{\bs{w}}
\newcommand{\bsx}{\bs{x}}
\newcommand{\bsy}{\bs{y}}
\newcommand{\bsz}{\bs{z}}
\newcommand{\bszero}{\bs{0}}
\newcommand{\bsalpha}{\bs{\alpha}}
\newcommand{\bsbeta}{\bs{\beta}}
\newcommand{\bsgamma}{\bs{\gamma}}
\newcommand{\bsdelta}{\bs{\delta}}
\newcommand{\bsepsilon}{\bs{\epsilon}}
\newcommand{\bsvarepsilon}{\bs{\varepsilon}}
\newcommand{\bszeta}{\bs{\zeta}}
\newcommand{\bseta}{\bs{\eta}}
\newcommand{\bstheta}{\bs{\theta}}
\newcommand{\bsvartheta}{\bs{\vartheta}}
\newcommand{\bsiota}{\bs{\iota}}
\newcommand{\bskappa}{\bs{\kappa}}
\newcommand{\bslambda}{\bs{\lambda}}
\newcommand{\bsmu}{\bs{\mu}}
\newcommand{\bsnu}{\bs{\nu}}
\newcommand{\bsxi}{\bs{\xi}}
\newcommand{\bspi}{\bs{\pi}}
\newcommand{\bsvarpi}{\bs{\varpi}}
\newcommand{\bsrho}{\bs{\rho}}
\newcommand{\bsvarrho}{\bs{\varrho}}
\newcommand{\bssigma}{\bs{\sigma}}
\newcommand{\bsvarsigma}{\bs{\varsigma}}
\newcommand{\bstau}{\bs{\tau}}
\newcommand{\bsupsilon}{\bs{\upsilon}}
\newcommand{\bsphi}{\bs{\phi}}
\newcommand{\bsvarphi}{\bs{\varphi}}
\newcommand{\bschi}{\bs{\chi}}
\newcommand{\bspsi}{\bs{\psi}}
\newcommand{\bsomega}{\bs{\omega}}
\newcommand{\bsGamma}{\bs{\Gamma}}
\newcommand{\bsDelta}{\bs{\Delta}}
\newcommand{\bsTheta}{\bs{\Theta}}
\newcommand{\bsLambda}{\bs{\Lambda}}
\newcommand{\bsXi}{\bs{\Xi}}
\newcommand{\bsPi}{\bs{\Pi}}
\newcommand{\bsSigma}{\bs{\Sigma}}
\newcommand{\bsUpsilon}{\bs{\Upsilon}}
\newcommand{\bsPhi}{\bs{\Phi}}
\newcommand{\bsPsi}{\bs{\Psi}}
\newcommand{\bsOmega}{\bs{\Omega}}

\newcommand{\mbf}[1]{\mathbf{#1}}
\newcommand{\mbfA}{\mbf{A}}
\newcommand{\mbfB}{\mbf{B}}
\newcommand{\mbfC}{\mbf{C}}
\newcommand{\mbfD}{\mbf{D}}
\newcommand{\mbfE}{\mbf{E}}
\newcommand{\mbfF}{\mbf{F}}
\newcommand{\mbfG}{\mbf{G}}
\newcommand{\mbfH}{\mbf{H}}
\newcommand{\mbfI}{\mbf{I}}
\newcommand{\mbfJ}{\mbf{J}}
\newcommand{\mbfK}{\mbf{K}}
\newcommand{\mbfL}{\mbf{L}}
\newcommand{\mbfM}{\mbf{M}}
\newcommand{\mbfN}{\mbf{N}}
\newcommand{\mbfO}{\mbf{O}}
\newcommand{\mbfP}{\mbf{P}}
\newcommand{\mbfQ}{\mbf{Q}}
\newcommand{\mbfR}{\mbf{R}}
\newcommand{\mbfS}{\mbf{S}}
\newcommand{\mbfT}{\mbf{T}}
\newcommand{\mbfU}{\mbf{U}}
\newcommand{\mbfV}{\mbf{V}}
\newcommand{\mbfW}{\mbf{W}}
\newcommand{\mbfX}{\mbf{X}}
\newcommand{\mbfY}{\mbf{Y}}
\newcommand{\mbfZ}{\mbf{Z}}

\newcommand{\mscr}[1]{\mathscr{#1}}
\newcommand{\msA}{\mscr{A}}
\newcommand{\msB}{\mscr{B}}
\newcommand{\msC}{\mscr{C}}
\newcommand{\msD}{\mscr{D}}
\newcommand{\msE}{\mscr{E}}
\newcommand{\msF}{\mscr{F}}
\newcommand{\msG}{\mscr{G}}
\newcommand{\msH}{\mscr{H}}
\newcommand{\msI}{\mscr{I}}
\newcommand{\msJ}{\mscr{J}}
\newcommand{\msK}{\mscr{K}}
\newcommand{\msL}{\mscr{L}}
\newcommand{\msM}{\mscr{M}}
\newcommand{\msN}{\mscr{N}}
\newcommand{\msO}{\mscr{O}}
\newcommand{\msP}{\mscr{P}}
\newcommand{\msQ}{\mscr{Q}}
\newcommand{\msR}{\mscr{R}}
\newcommand{\msS}{\mscr{S}}
\newcommand{\msT}{\mscr{T}}
\newcommand{\msU}{\mscr{U}}
\newcommand{\msV}{\mscr{V}}
\newcommand{\msW}{\mscr{W}}
\newcommand{\msX}{\mscr{X}}
\newcommand{\msY}{\mscr{Y}}
\newcommand{\msZ}{\mscr{Z}}

\newcommand{\mcal}[1]{\mathcal{#1}}
\newcommand{\mcA}{\mcal{A}}
\newcommand{\mcB}{\mcal{B}}
\newcommand{\mcC}{\mcal{C}}
\newcommand{\mcD}{\mcal{D}}
\newcommand{\mcE}{\mcal{E}}
\newcommand{\mcF}{\mcal{F}}
\newcommand{\mcG}{\mcal{G}}
\newcommand{\mcH}{\mcal{H}}
\newcommand{\mcI}{\mcal{I}}
\newcommand{\mcJ}{\mcal{J}}
\newcommand{\mcK}{\mcal{K}}
\newcommand{\mcL}{\mcal{L}}
\newcommand{\mcM}{\mcal{M}}
\newcommand{\mcN}{\mcal{N}}
\newcommand{\mcO}{\mcal{O}}
\newcommand{\mcP}{\mcal{P}}
\newcommand{\mcQ}{\mcal{Q}}
\newcommand{\mcR}{\mcal{R}}
\newcommand{\mcS}{\mcal{S}}
\newcommand{\mcT}{\mcal{T}}
\newcommand{\mcU}{\mcal{U}}
\newcommand{\mcV}{\mcal{V}}
\newcommand{\mcW}{\mcal{W}}
\newcommand{\mcX}{\mcal{X}}
\newcommand{\mcY}{\mcal{Y}}
\newcommand{\mcZ}{\mcal{Z}}
